\documentclass[final]{siamltex}
\usepackage{booktabs}
\usepackage[mathscr]{euscript}
\usepackage{color,amsmath,amssymb,mathrsfs}
\usepackage{graphicx}
\usepackage{algorithmic,algorithm}
\usepackage{tikz,pgfplots}
\usepackage{upgreek}
\usepackage{showkeys,cite}
\usepackage[normalem]{ulem}
\usepackage[pdfborder={0 0 0}, colorlinks=true, linkcolor=blue]{hyperref}

\usepackage{latexsym}

\hypersetup{pdfborder=0 0 0}

\makeatletter
\def\addcontentsline#1#2#3{%
  \addtocontents{#1}{\protect\contentsline{#2}{#3}{\thepage}}}
\long\def\addtocontents#1#2{%
  \protected@write\@auxout
    {\let\label\@gobble \let\index\@gobble \let\glossary\@gobble}%
    {\string\@writefile{#1}{#2}}}

\newcommand\tableofcontents{%
    \section*{
        \@mkboth{\noindent%
           \MakeUppercase\contentsname}{\MakeUppercase\contentsname}}%
    \@starttoc{toc}%
    }
\newcommand*\l@section[2]{%
  \ifnum \c@tocdepth >\z@
    \addpenalty\@secpenalty
    \addvspace{1.0em \@plus\p@}%
    \setlength\@tempdima{1.5em}%
    \begingroup
      \parindent \z@ \rightskip \@pnumwidth
      \parfillskip -\@pnumwidth
      \leavevmode \bfseries
      \advance\leftskip\@tempdima
      \hskip -\leftskip
      #1\nobreak\hfil \nobreak\hb@xt@\@pnumwidth{\hss #2}\par
    \endgroup
  \fi}
\newcommand*\l@subsection{\@dottedtocline{2}{1.5em}{2.3em}}
\newcommand*\l@subsubsection{\@dottedtocline{3}{3.8em}{3.2em}}
\newcommand*\l@paragraph{\@dottedtocline{4}{7.0em}{4.1em}}
\newcommand*\l@subparagraph{\@dottedtocline{5}{10em}{5em}}
\makeatother

\newcommand{\llangle}{\left\langle\kern-2.1\nulldelimiterspace\left\langle}
\newcommand{\rrangle}{\right\rangle\kern-2.1\nulldelimiterspace\right\rangle}

\newcommand{\MATLAB}{\textsc{Matlab}}

\renewcommand{\vec}[1]{{\mathchoice
                     {\mbox{\boldmath$\displaystyle{#1}$}}
                     {\mbox{\boldmath$\textstyle{#1}$}}
                     {\mbox{\boldmath$\scriptstyle{#1}$}}
                     {\mbox{\boldmath$\scriptscriptstyle{#1}$}}}}
\newcommand{\mat}[1]{\mathbf{{#1}}}

\newcommand{\trace}{\mathsf{tr}}
\newcommand{\ave}[1]{\mathsf{E}\left\{ {#1} \right\}}
\newcommand{\var}[1]{\mathsf{Var}\left\{ {#1} \right\}}
\newcommand{\cov}[2]{\mathsf{Cov}\left\{ {#1},{#2} \right\}}

\newcommand{\norm}[1]{\left\| {#1} \right\|}
\renewcommand{\span}{\mathsf{span}}

\renewcommand{\diag}{\mathsf{diag}}

\newcommand{\eps}{\varepsilon}

\newcommand{\Nt}{{N_t}}
\newcommand{\Ns}{{N_s}}
\newcommand{\Nm}{{n}}
\newcommand{\Ntau}{{N_{\tau}}}
\newcommand{\Ntr}{{N_{{\mathsf{tr}}}}}

\newcommand{\uu}{\vec{\bar{u}}}
\newcommand{\dd}{\vec{\bar{d}}}
\newcommand{\BB}{\matrix{B}}
\newcommand{\WW}{\matrix{{W}}}
\renewcommand{\Xi}{X^{-1}}
\newcommand{\ip}[2]{\left\langle{#1}, {#2}\right\rangle}
\newcommand{\mip}[2]{\left\langle{#1}, {#2}\right\rangle_{\!\scriptscriptstyle{\text{M}}}}

\newcommand{\vv}{\ensuremath{\vec{v}}}

\renewcommand{\matrix}[1]{\mathbf{#1}}
\renewcommand{\H}{\ensuremath{\matrix{H}}}
\newcommand{\M}{{\ensuremath{\matrix{M}}}}

\newcommand{\I}{{\ensuremath{\matrix{I}}}}
\newcommand{\A}{{\ensuremath{\matrix{A}}}}

\newcommand{\V}{{\ensuremath{\matrix{V}}}}

\renewcommand{\L}{{\ensuremath{\matrix{L}}}}
\newcommand{\F}{{\ensuremath{\matrix{F}}}}

\newcommand{\FT}{\matrix{\tilde{F}}}
\newcommand{\HT}{\matrix{\tilde{H}}_{\text{misfit}}}
\newcommand{\HTr}{\matrix{\tilde{H}}_{\text{misfit},r}}

\newcommand{\HM}{\matrix{H}_{\text{misfit}}}

\newcommand{\obs}{\vec{d}}  
\newcommand{\ipar}{m}
\newcommand{\dpar}{\vec{m}}

\newcommand{\iparpost}{\ipar_\text{post}}

\newcommand{\dparpost}{\dpar_\text{post}}

\newcommand{\xx}{\ensuremath{\boldsymbol x}}

\newcommand{\eqnlab}[1]{\label{eq:#1}}

\definecolor{pacificorange}{cmyk}{.15,.45,1,0} 
\definecolor{pacificgray}{cmyk}{0,.15,.35,.60}
\definecolor{pacificlgray}{cmyk}{0,0,.2,.4}
\definecolor{pacificcream}{cmyk}{.05,.05,.15,0}
\definecolor{deepyellow}{cmyk}{0,.17,.80,0}
\definecolor{lightblue}{cmyk}{.49,.01,0,0}
\definecolor{lightbrown}{cmyk}{.09,.15,.34,0}
\definecolor{deepviolet}{cmyk}{.79,1,0,.15}
\definecolor{deeporange}{cmyk}{0,.59,1,.18}
\definecolor{dustyred}{cmyk}{0,.7,.45,.4}
\definecolor{grassgreen}{RGB}{92,135,39}
\definecolor{pacificblue}{RGB}{59,110,143}
\definecolor{pacificgreen}{cmyk}{.15,0,.45,.30}
\definecolor{deepblue}{cmyk}{1,.57,0,.2}
\definecolor{turquoise}{cmyk}{.43,0,.24,0}
\definecolor{green}{rgb}{0,0.65,0}

\newcommand{\bs}[1]{\boldsymbol{#1}}

\newcommand{\equaldef}{{:=}}

\newcommand{\mc}[1]{\mathcal{#1}}
\newcommand{\nor}[1]{\left\| #1 \right\|}

\newcommand{\R}{\mathbb{R}}
\newcommand{\half} {\ensuremath{\frac{1}{2}}}
\newcommand{\D}{\mathcal{D}}
\newcommand{\Acal}{\mc{A}}
\newcommand{\borel}{\mathfrak{B}}
\newcommand{\hilb}{\mathscr{H}}

\newcommand{\C}{\mc{C}}
\newcommand{\Grad} {\ensuremath{\nabla}}

\newcommand{\LRp}[1]{\left( #1 \right)}

\newcommand{\LRc}[1]{\left\{ #1 \right\}}
\newcommand{\GM}[2]{\mc{N}\left( #1, #2 \right)}
\newcommand{\del}{\partial}
\newcommand{\delj}[1]{\frac{\partial{#1}}{\partial {w_j}}}

\newcommand{\obsop}{\mathcal{B}}

\newcommand{\prior}  { \pi_{\rm{\scriptscriptstyle{prior}}}   }
\newcommand{\like}{ \pi_{\rm{\scriptscriptstyle{like}}} }

\newcommand{\postm}{\mu_{\rm{post}}}
\newcommand{\priorm}{\mu_0}
\newcommand{\ncov} {\bs \Gamma_{\rm{noise}} }

\newcommand{\prcov} {\bs \Gamma_{\rm{prior}}    }
\newcommand{\postcov} {\bs \Gamma_{\rm{post}}    }

\newcommand{\Vr} {\matrix{V}_r}

\newcommand{\Dr} {\matrix{D}_r}

\newcommand{\ML}{\M_l}
\newcommand{\Mt}{\tilde{\M}}
\newcommand{\iFF}{\mathcal{F}}
\newcommand{\ff}{\mathcal{F}}     
\newcommand{\FF}{\matrix{F}}

\newcommand{\sol}{\matrix{S}}

\newcommand{\adjMacroMM}[1]{{#1}^*}
\newcommand{\adjMacroME}[1]{{#1}^*}

\newcommand{\Aadj}{\adjMacroMM{\A}}

\newcommand{\iFFadj}{\adjMacroME{\iFF}}

\newcommand{\Fadj}{\adjMacroME{\F}}

\newcommand{\wopt}{\vec{w}_{\text{opt}}}

\def\addressices{Institute for Computational Engineering \& Sciences, The
  University of Texas at Austin, Austin, TX, USA}
\def\addressgeo{Department of Geological Sciences, The University of
  Texas at Austin, Austin, TX, USA}
\def\addressmech{Department of Mechanical Engineering, The
  University of Texas at Austin, Austin, TX, USA}

\begin{document}

\author{Alen Alexanderian\footnotemark[2] \and Noemi Petra\footnotemark[2]
  \and Georg Stadler\footnotemark[2] \and Omar~Ghattas{\footnotemark[2]~\footnotemark[3]~\footnotemark[4]}}
\renewcommand{\thefootnote}{\fnsymbol{footnote}}
\footnotetext[2]{\addressices}
\footnotetext[3]{\addressmech}
\footnotetext[4]{\addressgeo}
\renewcommand{\thefootnote}{\arabic{footnote}}
\title{A-Optimal design of experiments for infinite-dimensional
  Bayesian linear inverse problems with regularized
  $\ell_0$-sparsification\thanks{This work was partially supported by
    NSF grant ARC-0941678; DOE grants 
DE-FC02-13ER26128, %
DE-SC0010518 %
DE-FC02-11ER26052, %
DE-11018096, %
and 
DE-FG02-09ER25914; %
and AFOSR grant  FA9550-12-1-0484.}} 

\maketitle

\begin{abstract}
We present an efficient method for computing A-optimal experimental
designs 
for infinite-dimensional Bayesian linear inverse problems
governed by partial differential equations
(PDEs). Specifically, we 
address the 
problem of optimizing the location
of sensors (at which observational data are collected) 
to minimize the uncertainty in the parameters estimated by solving the
inverse problem, where the uncertainty is expressed by the trace of
the posterior covariance.
Computing optimal experimental designs (OEDs) is particularly
challenging for inverse problems governed by computationally expensive
PDE models with infinite-dimensional (or, after discretization,
high-dimensional) parameters.  To alleviate the computational cost, we
exploit the problem structure
and build a low-rank approximation of the parameter-to-observable map,
preconditioned with the square root of the prior covariance
operator. The availability of this low-rank surrogate, relieves our
method from expensive PDE solves when evaluating the optimal
experimental design objective function, i.e., the trace of the
posterior covariance, and its derivatives. Moreover, we employ a
randomized trace estimator for efficient evaluation of the OED
objective function.  We control the sparsity of the sensor
configuration by employing a sequence of penalty functions that
successively approximate the $\ell_0$-``norm''; this results in binary
designs that characterize optimal sensor locations.
We present numerical results for inference of the initial condition
from spatio-temporal observations in a time-dependent
advection-diffusion problem in two and three space dimensions.  We
find that an optimal design can be computed at a cost, measured in
number of forward PDE solves, that is independent of the parameter and sensor
dimensions. Moreover, the numerical optimization problem for finding
the optimal design can be solved in a number of interior-point
quasi-Newton iterations that is insensitive to the parameter and
sensor dimensions. We demonstrate numerically that $\ell_0$-sparsified
experimental designs obtained via a continuation method outperform
$\ell_1$-sparsified designs. 
\end{abstract}



\section{Introduction}\label{sec:intro}

Recent advances in theory~\cite{Stuart10} and numerical algorithms
(e.g., \cite{Bui-ThanhGhattasMartinEtAl13}) are enabling efficient
solution of infinite-dimensional Bayesian inverse
problems. This opens the door to consideration of the upstream
question: how do we place sensors to \emph{optimally} infer model
parameters for large-scale problems? Here we present an efficient method
for such optimal experimental design (OED) problems.
Specifically, we consider Bayesian linear inverse problems governed by
PDEs whose solution is the posterior probability law for a parameter
field.  The numerical solution of such inverse problems is challenging
due to the infinite (or, when discretized, large) dimension of the
parameters, ill-posedness of the inverse problem, and
expensive-to-compute PDE models.
The Bayesian inverse problem is by itself very challenging, but it is
merely a subproblem within the OED problem, and must be solved
repeatedly when using conventional OED methods. 
Hence, it is essential to make maximum use of the problem structure to
realize efficient algorithms that are \emph{scalable}, i.e., their
performance---measured in number of governing (forward) PDE solves---is
independent of the discretized parameter and sensor dimensions, and
the discretization of the governing PDE.

When formulating an OED problem, a basic question is the precise
meaning of what constitutes the \emph{design}.  The
present work concerns computation of optimal sensor locations where
observational data will be collected. A subsequent question concerns the
definition of an \emph{optimal} design, which leads to the choice of
the design criterion. For Bayesian inverse problems, a natural choice
is to seek a design that minimizes the average posterior variance of
the inversion parameters, leading to the Bayesian A-optimal design
criterion.

Standard references for optimal experimental design
include~\cite{Ucinski05, AtkinsonDonev92, Pukelsheim93,
  Pazman86,BockKoerkelSchloeder13}.
While most of the classical texts concern problems with small or
moderate parameter dimension and focus mainly on well-posed inverse
problems, there has been recent interest in optimal design for
large-scale ill-posed
linear~\cite{HaberHoreshTenorio08,HaberMagnantLuceroEtAl12} and
nonlinear~\cite{HoreshHaberTenorio10,HaberHoreshTenorio10} inverse
problems. The numerical methods in the present paper are closest to
those in~\cite{HaberHoreshTenorio08,HaberMagnantLuceroEtAl12}, where
the authors consider finite-dimensional linear inverse problems, and
develop a framework to control the mean square error of the
regularized Tikhonov estimates.  This leads to a design criterion that
seeks to minimize the sum of the estimation bias and the variability
of the estimator around its mean, which can be related to A-optimal
designs.  These contributions use $\ell_1$-penalties to control the
sparsity of the design and, in~\cite{HaberMagnantLuceroEtAl12}, a
low-rank singular value decomposition (SVD) of the
parameter-to-observable map is used; the present paper builds on and
extends both of these ideas.
Further recent work that employs a Bayesian formulation 
includes~\cite{HuanMarzouk13,HuanMarzouk12}, where the authors use a
decision theoretic design criterion, generalized polynomial chaos
surrogates, and stochastic optimization to tackle nonlinear inverse 
problems, albeit in low to moderate parameter dimension.

In this paper, we 
devise scalable numerical methods for computing A-optimal designs for
infinite-dimensional Bayesian linear inverse problems governed by
(time-dependent) PDEs.  As suggested
in~\cite{HaberMagnantLuceroEtAl12}, having a low-rank SVD
\emph{surrogate} of the parameter-to-observable map relieves the
OED method of repeated PDE solves. In the Bayesian
context, we can improve on this idea and further exploit problem
structure; namely, we
construct a low-rank SVD representation of the parameter-to-observable
map \emph{preconditioned} by the square root of the prior covariance
operator~\cite{Bui-ThanhGhattasMartinEtAl13,FlathWilcoxAkcelikEtAl11}. This
preconditioning amounts to filtering through the prior the information
gained from the data about the model parameters. In the case of
smoothing priors usually used in infinite-dimensional problems, this
preconditioning results in faster spectral decay and thus allows for a
more efficient low-rank approximation. The remaining steps in the
solution of the OED problem use this low-rank surrogate and thus do
not require additional PDE solves.  As a result of a consistent
discretization of the problem, i.e., one that respects the
infinite-dimensional Hilbert-space structure, the numerical rank of
the prior-preconditioned parameter-to-observable map is bounded with
respect to the discretized parameter dimension.

We consider a finite number of candidate locations
for the placement of sensors; the optimal configuration is a sparse
subset of these locations.  To each candidate sensor location we
assign a non-negative number that weights the observation from that
sensor. Finding an optimal design then amounts to choosing an optimal
weight vector; a weight of 0 indicates absence of a sensor and a
weight of 1 corresponds to a sensor being placed at that location. For
computational convenience, we allow the weights to take on any value
in $[0,1]$, and use a sparsifying penalty to control the number of
nonzero weights, and thus, the number of allocated sensors.  One
option for such a penalty is the $\ell_1$-norm; see,
e.g.,~\cite{HaberHoreshTenorio08, HoreshHaberTenorio10,
  HaberMagnantLuceroEtAl12}.  This approach, however, does not %
lead to a binary (i.e., 0--1) design.  Motivated by continuation methods
used in topology optimization~\cite{Bendsoe95,BendsoeSigmund03}, we
propose to solve a sequence of OED problems with penalty functions
that successively approximate the $\ell_0$-``norm''.  This, in
contrast to an $\ell_1$-penalty approach, does result in a binary
design.  As a test problem for our OED method, we consider a forward
problem in the form of a time-dependent advection diffusion model, in
which we infer the probability law of the initial condition. For this
problem, we demonstrate the success of our continuation method, and
show that the weights found by this continuation approach improve over
optimal designs obtained via an $\ell_1$-penalty approach.

For typical infinite-dimensional Bayesian inverse problems, the
performance of our OED method is insensitive to the number of
candidate sensor locations. This is due to the fact that although the
dimension of the observations increases with the number of candidate
locations, the amount of independent information that can be gained
from nearby sensors is typically limited. Thus, increasing the number
of sensors beyond a certain point does not significantly increase the
numerical rank of the prior-preconditioned parameter-to-observable
map. As a consequence, the number of PDE solves required to compute a
low-rank SVD surrogate for the (prior-preconditioned)
parameter-to-observable map is bounded as the number of candidate
sensor locations increases.  Moreover, in our computational results we
find that the numerical optimization problem to compute the optimal
design can be solved in a number of (quasi-Newton) iterations that is
independent of the number of candidate locations.

The large-scale nature of the Bayesian inverse problems we target
necessitates the use of randomized methods in linear algebra.  In
particular, we utilize randomized trace
estimators~\cite{Hutchinson90,AvronToledo11} to estimate the trace of
the posterior covariance operator, and randomized
SVD~\cite{HalkoMartinssonTropp11} to compute a low-rank surrogate of
the prior-preconditioned parameter-to-observable map.  Moreover, for
computing the application of matrix square roots, as needed in our
method, we employ matrix-free iterative
methods~\cite{ChenAnitescuSaad11}.

The structure of this paper is as follows. After presenting the
requisite background material in Section~\ref{sec:background}, we
formulate the optimal design problem in infinite dimensions in
Section~\ref{sec:oed}. Then, we detail the components of our OED
method in Section~\ref{sec:oed-algorithms}. Section~\ref{sec:model}
provides a description of our model problem, namely the inference of the
initial condition in a time-dependent advection-diffusion equation.
We present a comprehensive numerical study in
Section~\ref{sec:numerics}.  Finally, in Section~\ref{sec:conc}, we
draw conclusions, and discuss limitations and possible extensions of
our method.

\section{Background} \label{sec:background}
In this section, we provide the background material required for the
formulation and numerical solution of optimal experimental design problems 
in the context of infinite-dimensional Bayesian inverse
problems.
In Section~\ref{sec:bayes}, we present
the Bayesian inverse problem in an infinite-dimensional Hilbert space
setting, adopting the framework in~\cite{Stuart10}. 
In Section~\ref{sec:bayes-disc}, we 
describe a discretization that is 
consistent with the infinite-dimensional inference problem 
formulation; this presentation follows~\cite{Bui-ThanhGhattasMartinEtAl13}.
Finally, in Section~\ref{sec:rand}, we briefly comment on the randomized SVD and randomized trace
estimators, which are used in our numerical method.

\subsection{Bayesian inversion in Hilbert spaces} \label{sec:bayes}
We begin our discussion by first considering a deterministic inverse problem.
Given finite-dimensional observations $\obs \in \R^q$, we seek the model parameter $\ipar$ that solves
\begin{equation}\label{eq:optpb}
  \min_{\ipar \in \hilb} \mc{J}\LRp{\ipar} \equaldef
  \half \nor{\ff\LRp{\ipar} - \obs 
  }^2_{\bs{\Gamma}} + \mathcal{R}(m).
\end{equation}
The function $\ff: \hilb \to \mathbb{R}^q$ is the parameter-to-observable map 
and $\mathcal{R}$ denotes a regularization term.  
In the applications targeted in this paper $\hilb = L^2(\D)$, 
where $\D \subset \R^d$ is a bounded domain (with $d = 2, 3$)  and an evaluation of $\ff$ 
involves the solution of a PDE, followed by the application of an observation operator.
Note that the solution of a deterministic inverse problem 
can be thought of as a \emph{point estimate} of $\ipar$. 
To obtain a full probabilistic description of the parameter $\ipar$, 
we are led to a Bayesian formulation of the problem, 
whose solution is a posterior probability law for $m$. 

In this paper, we consider a parameter $m$ which is modeled as
a random-field (random function). 
To be precise, letting $(\Omega, \Sigma, \mathsf{P})$ be an 
appropriate probability space, $\ipar:\D \times \Omega \to \R$ is a function such that
for each $\vec{x} \in \D$, $\ipar(\vec{x}, \cdot)$, is a real-valued random variable; 
thus, we can view $\ipar$ as an indexed collection of random 
variables, $\{ m(\vec{x}) \}_{\vec{x} \in \D}$, where, following the
common practice, we suppress the dependence on $\omega$.
On the other hand, for each $\omega \in \Omega$, $\ipar(\cdot, \omega):\D \to \R$
is a real-valued function. We consider the case where $\ipar(\cdot, \omega) \in \hilb$, and
thus, we can also view $m$ as a random-variable, $m:(\Omega, \Sigma,\mathsf{P}) \to \big(\hilb, \borel(\hilb)\big)$, 
where $\borel(\hilb)$ denotes the Borel $\sigma$-algebra on $\hilb$. Recall that the law of $m$
is a probability measure $\mu$ on $\big(\hilb, \borel(\hilb)\big)$ given by 
$\mu(E) = \mathsf{P}( m \in E)$ for $E \in \borel(\hilb)$.

The infinite-dimensional Bayesian inverse problem can then be
formulated as using observations to update our knowledge of the law of $\ipar$, as a probability 
measure on $(\hilb, \borel(\hilb))$. Since, 
in contrast to the finite-dimensional case, 
there is no Lebesgue measure on $\hilb$,
the infinite-dimensional Bayes formula is given by
\begin{equation} \label{equ:bayes-abstract}
   \frac{d\postm}{d\priorm} \propto \like(\obs | \ipar).
\end{equation}
Here, $\frac{d\postm}{d\priorm}$ denotes the Radon-Nikodym 
derivative~\cite{Williams1991} of the posterior measure $\postm$ with respect to $\priorm$, and
$\like(\obs | \ipar)$ denotes the data likelihood. 
Conditions under which the posterior 
measure is well defined and~\eqref{equ:bayes-abstract} holds are given in detail 
in~\cite{Stuart10}. 
Note that we consider a finite-dimensional observation vector, 
motivated by the fact that in practice data are available only at 
a finite number of sensor locations and a finite number of points in time.

In the present work, we consider the Gaussian-linear case, i.e., the
parameter-to-observable map $\iFF$ is linear. Moreover, we assume
an additive noise model, $\obs = \ff\ipar + \vec{\eta}$, 
where $\vec{\eta} \sim \GM{\vec{0}}{\ncov}$ is a centered Gaussian on $\R^q$; the latter implies 
\begin{equation}\label{equ:likelihood}
\like(\obs | \ipar) \propto \exp\Big\{-\half(\iFF\ipar - \vec{d})^T\ncov^{-1}(\iFF\ipar - \vec{d})\Big\}.
\end{equation}
We use a Gaussian prior, $\mu_0 = \GM{\ipar_0}{\C_0}$, 
where $\ipar_0 \in \hilb$ is sufficiently regular, and $\C_0$ is an appropriate 
covariance operator, i.e., $\C_0$ must be symmetric, positive, and
of trace-class. We define
the covariance operator as the inverse of an elliptic differential
operator. A common alternative choice for statistical inverse problems is
to specify a covariance function between any two spatial points, which
results in a dense covariance
matrix. For large-scale problems, however, the construction
and ``inversion'' of such a dense covariance matrix can be
infeasible. 
On the contrary, specifying the covariance as the inverse of an
elliptic differential operator allows to build on existing fast
solvers for elliptic equations.
As detailed in~\cite{Stuart10,Bui-ThanhGhattasMartinEtAl13}, 
the PDE solution operator used as covariance operator $\C_0$ must be 
sufficiently smoothing and have bounded Green's functions. 
For example, the biharmonic operator has bounded Green's functions in
two and three space dimensions.  Therefore, we choose $\C_0=\mc{A}^{-2}$, with
$\mc{A}$ a Laplacian-like operator in the sense of Assumption 2.9 in \cite{Stuart10}.
This choice also allows efficient 
applications of the square root operator $\C_0^{1/2}=\mc{A}^{-1}$,
as required below.
The elliptic PDE corresponding to $\mc A$ written in weak
form is as follows: For $s\in \hilb = L^2(\D)$, the solution $\ipar=\mc A^{-1}s$ satisfies
\begin{equation}
\eqnlab{Wspace}
    \int_\D \alpha \nabla \ipar \cdot \nabla p + \beta \ipar p \,d\xx = \int_\D
    sp\,d\xx, \quad \text{ for all } p \in H^1(\D),
\end{equation}
with $\alpha, \beta > 0$ controlling the variation and the correlation length.
Due to the present choice of the prior and the noise model, and the linearity of $\iFF$, the posterior measure 
is a Gaussian, $\GM{\iparpost}{\C_\text{post}}$ with~\cite[Section 6.4]{Stuart10}, 
\begin{equation}\label{equ:mean-cov}
\C_\text{post} = (\iFFadj \bs{\Gamma}_\text{noise}^{-1} \iFF + \C_0^{-1})^{-1}, 
\qquad 
\iparpost = \C_\text{post}(\iFFadj\ncov^{-1}\obs + \C_0^{-1}\ipar_0),
\end{equation}
where $\iFFadj:\R^q \to \hilb$ is the adjoint of $\iFF$.

\subsection{Discretization of the infinite-dimensional Bayesian inverse problem}
\label{sec:bayes-disc}
In this section, we describe the 
discretization of the Bayesian inverse problem~\eqref{equ:bayes-abstract} in the Gaussian linear 
case. %
We consider a finite-dimensional 
subspace $V_h \subset L^2(\D)$ given by 
$V_h = \span\{\phi_1, \ldots, \phi_n\}$, 
where $\LRc{\phi_j}_{j=1}^n$ are continuous Lagrange nodal basis functions.
Given $\ipar_h \in V_h$, we denote by $\vec{m}$, the vector of its 
coordinates in $V_h$; i.e., for $\ipar_h = \sum_{j=1}^nm_j\phi_j$, we 
have $\vec{m} = \LRp{m_1,\hdots,m_n}^T$.
After this discretization, we replace the task of inferring the parameter 
$\ipar \in L^2(\D)$ with that of inferring the 
coefficients for the finite-element 
approximation $\ipar_h$ of $\ipar$.

Following~\cite{Bui-ThanhGhattasMartinEtAl13}, 
we state the finite-dimensional Bayesian inverse problem such that it is
consistent with the corresponding inference problem in $L^2(\D)$.
Consequently, we work in $\R^n$, with the weighted inner product, $\mip{\cdot\,}{\cdot}$ 
given by $\mip{\vec{x}}{\vec{y}} = \ip{\M\vec{x}}{\vec{y}}$, where $\ip{\cdot}{\cdot}$ denotes 
the Euclidean inner product, and $\M$ is the (symmetric positive definite) finite-element 
mass matrix.
It is convenient to introduce the notation $\R^n_{\M}$ for $\R^n$
when endowed with the $\mip{\cdot\,}{\cdot}$ inner product. Note that the mapping
$\ipar_h \mapsto \vec{m}$ is a Hilbert-space isomorphism between
$V_h$ (with $L^2$-inner product) and $\R^n_{\M}$.
For a linear operator $\matrix{A}:\R^n_\M \to \R^n_\M$, the adjoint
operator is given by $\matrix{A}^* = \M^{-1} \matrix{A}^T \M$.
A linear operator $\matrix{A}$ on $\R^n_\M$ is self-adjoint if $\matrix{A} = \matrix{A}^*$; 
for convenience, we refer to such operators as $\M$-symmetric. 
In the sequel, we will encounter linear mappings  
$\A_1 : \R^n_{\M} \rightarrow \R^q$ and $\A_2 : \R^r \rightarrow \R^n_{\M}$, 
where $\R^q$ and $ \R^r$ are endowed with the Euclidean inner product;
the corresponding adjoints are given by~\cite{Bui-ThanhGhattasMartinEtAl13},
$\Aadj_1 = \M^{-1} \A_1^T$ and $\Aadj_2 = \A_2^T \M$.

For the discretized problem, 
the density for the prior (as a measure over the space $\R^n_{\M}$) is characterized by
\begin{equation}
\label{eq:prior_pdf}
\prior(\dpar) \:\propto\:
 \exp\left\{
- \frac 12 \left\| \A (\dpar - \dpar_{0}) \right\|^2_{\M}
\right\}, 
\end{equation}
with $\A = \M^{-1}\matrix{L}$, where $\matrix{L} = \alpha \matrix{K} + \beta \matrix{M}$
and $\matrix{K}$ is the finite-element stiffness matrix.
Note that $\matrix{A}$ is $\M$-symmetric; moreover, it follows from 
the above definition that $\matrix{\Gamma}_{\text{prior}} = \A^{-2}$.
Subsequently, the posterior is a Gaussian $\GM{\dpar_\text{post}}{\postcov}$
with
\[
   \dpar_\text{post} = \matrix{\Gamma}_{\text{post}} \left(
   \Fadj\matrix{\Gamma}_{\text{noise}}^{-1} \obs +
   \matrix{\Gamma}^{-1}_{\text{prior}} \dpar_0\right), 
   \qquad  
    \matrix{\Gamma}_{\text{post}} = \left(
   \Fadj\matrix{\Gamma}_{\text{noise}}^{-1}\F+\matrix{\Gamma}^{-1}_{\text{prior}}\right)^{-1},
\]
where $\F: \R^n_{\M} \rightarrow \R^q$ is the discretization of the
parameter-to-observable map $\iFF$.

Next,  we summarize how to draw samples from and to compute the variance 
of the discretized posterior measure; here we follow~\cite{Bui-ThanhGhattasMartinEtAl13},
where details of these operations are provided.
For a posterior measure $\GM{\dpar_\text{post}}{\postcov}$ 
generating samples requires
a decomposition of $\postcov$ in the form, $\postcov = \matrix{Q}\matrix{Q}^*$.
Then, to generate a sample $\vec{\nu}$ 
from $\GM{\dpar_\text{post}}{\postcov}$, we draw a realization $\vec{z}$ from $\GM{0}{\matrix{I}}$ 
and compute $\vec{\nu}$ as %
\begin{equation} \label{equ:sample}
   \vec{\nu} = \dpar_\text{post} + \matrix{Q} \M^{-1/2} \vec{z}.
\end{equation}
The discretized covariance is given by
$\cov{m_i}{m_j} = \vec{e}_i^T \postcov \M^{-1} \vec{e}_j$, for $i, j = 1, \ldots, n$. In 
particular, 
\begin{equation}\label{equ:postvar}
\var{m_i} =  \vec{e}_i^T \postcov \M^{-1} \vec{e}_i, \quad i = 1, \ldots, n,
\end{equation}
where $\var{m_i}$ denotes the variance of $m_i$.
\subsection{Randomized linear algebra algorithms}\label{sec:rand}
One of the major components of our method is
a low-rank SVD surrogate for the prior-preconditioned
parameter-to-observable map. To compute such low-rank
surrogates, we use a randomized SVD~\cite{HalkoMartinssonTropp11}. This choice is motivated by
the flexibility and the robustness of the method, and by the fact that,
as opposed to Krylov subspace methods, it
only requires independent matrix-vector products. This aspect is
particularly useful for the large-scale problems we target, in which
the matrix-vector applications involve expensive PDE
solves; see also~\cite{Bui-ThanhBursteddeGhattasEtAl12}.
Randomized 
SVD methods can be made very accurate with negligible probability of 
failure~\cite{HalkoMartinssonTropp11}.

Computing A-optimal designs requires minimizing the trace of large
dense covariance matrices, which, in our target problems, 
usually have a rapidly decaying spectrum and the eigenvalues are 
clustered around 0. For such matrices, which are defined 
implicitly through their applications to vectors, randomized trace 
estimators provide a reasonably accurate approximation of the 
trace with a small number of random 
vectors (see e.g.,~\cite{Roosta-KhorasaniAscher13}).
These estimators involve only matrix-vector products, 
which makes them suitable for large scale problems. 
In particular, randomized trace estimators approximate the trace of a
matrix $\matrix{A} \in \R^{n \times n}$, via Monte-Carlo estimates of the
form $\trace(\matrix{A}) \approx \frac{1}{\Ntr} \sum_{i = 1}^\Ntr [\vec{z}^{(i)}]^T \matrix{A} \vec{z}^{(i)}$, 
where the trial vectors $\vec{z}^{(i)}$ are random $n$-vectors.
A well-known example is the Hutchinson estimator~\cite{Hutchinson90}, 
which uses random vectors $\vec{z}^{(i)}$ with $\pm 1$ entries,
each with a probability of $1/2$.
Another possibility, used in this paper, is the Gaussian trace 
estimator, which uses Gaussian random vectors with independent and identically distributed (i.i.d.) standard normal entries. For a description and analysis of different trace estimators,
we refer to~\cite{AvronToledo11}.

\section{A-optimal design of experiments for infinite-dimensional 
Bayesian linear inverse problems} \label{sec:oed}
In this section, we formulate the A-optimal design problem for 
infinite-dimensional Bayesian linear inverse problems.
The extension of the A-optimal design criterion to the infinite-dimensional
setting is described in Section~\ref{sec:oed-formulation-hilbert}. 
In Section~\ref{sec:oed-formulation-design}, we specify the mathematical 
definition of a \emph{design}, and describe how the design 
is introduced in the Bayesian inverse problem.
Finally, in Section~\ref{sec:oed-formulation-discrete}, we formulate the resulting 
OED optimization problem. 

\subsection{A-optimal design in Hilbert spaces}\label{sec:oed-formulation-hilbert}
In a finite-dimensional inference problem, an A-optimal design minimizes
the average posterior variance of the inference parameters~\cite{Ucinski05}. In the linear case, 
this is accomplished by minimizing the trace of the posterior \emph{covariance matrix}. 
Since the present work concerns inference problems with a 
random field as the inference parameter,
we first extend the notion of an
A-optimal design to the infinite-dimensional Hilbert space setting.

The average posterior variance of $\ipar$ over the
physical domain $\D$ is given by 
\[
\frac{1}{|\D|} \int_\D c_\text{post}(\vec{x},\vec{x}) \, d\vec{x}, 
\]
where $|\D|$ denotes the Lebesgue measure of the domain $\D$, and
$c_\text{post}$ is the \emph{covariance function} of $\ipar$:
\[
c_\text{post}(\vec{x},\vec{y}) = 
\ave{\Big(\ipar(\vec{x}) - \iparpost(\vec{x})\Big)\Big(\ipar(\vec{y}) - \iparpost(\vec{y})\Big)},
\qquad \vec{x}, \vec{y} \in \D,
\]
where $\ave{\cdot}$ denotes the expectation operator. Note that $c_\text{post}$
is related to the \emph{covariance operator} $\C_\text{post}$, i.e.,
\[
    [\C_\text{post} u](\vec{x}) = \int_\D c_\text{post}(\vec{x},\vec{y})u(\vec{y}) \, d\vec{y},
    \quad
    u \in L^2(\D).
\]
The covariance operator $\C_\text{post}$ is positive, symmetric, and of trace-class, 
and thus has real eigenvalues, $\{\lambda_i\}_{i=1}^\infty$, and a complete orthonormal
set of eigenvectors, $\{e_i\}_{i=1}^\infty$. Mercer's Theorem~\cite{Mercer1909,Lax02} then
provides a decomposition of the covariance function $c_{\text{post}}$
from the spectral decomposition of $\C_{\text{post}}$ through
$
   c_\text{post}(\vec{x},\vec{y}) = \sum_j \lambda_j e_j(\vec{x})e_j(\vec{y}), 
$
where the convergence of the infinite sum is uniform and absolute in $\D \times \D$. 
From this representation, one obtains,
\[
   \int_{\D} c_\text{post}(\vec{x},\vec{x})\,d\vec{x} = \sum_j \lambda_j = \trace(\C_\text{post}).
\]
Thus, we formulate the A-optimal design problem as that of minimizing
the trace of the covariance \emph{operator}, $\trace(\C_\text{post})$.
This shows that the definition of an A-optimal design in finite
dimensions, namely the minimization of the average variance of the
estimates, extends naturally to infinite-dimensional Bayesian inverse
problems involving a random-field as the inversion parameter.  We
emphasize that the trace of the posterior covariance operator in the
A-optimal design criterion is well defined due to the proper choice of
the prior measure for infinite-dimensional inference problems.

\subsection{Introducing the design in the Bayesian inverse problem} \label{sec:oed-formulation-design}
Let us first specify the notion of a design in the context of our target
applications, namely optimal sensor placement.  We use a
finite-dimensional design space, that is, we fix a set of points
$\vec{x}_i$, $i = 1, \ldots, \Ns$ as the set of candidate sensor
locations, and associate to each $\vec{x}_i$ a non-negative weight
$w_i \in \R$; see
also~\cite{Ucinski05,HaberHoreshTenorio08,HaberMagnantLuceroEtAl12}.
The OED problem is then formulated as an optimization problem over the
weight vector, $\vec{w} = (w_1, \ldots, w_\Ns)^T$.  The
points $\vec{x}_i$, $i = 1, \ldots, \Ns$ can be thought of as a
discretization of the \emph{sensor domain}, which is a subset of $\D$.

The design weights can have different interpretations. For example, in
classical formulations such as in~\cite{Ucinski05}, $w_i$ define a
probability mass function, i.e., $w_i \geq 0$ and $\sum w_i = 1$; one
may then interpret large weights as ones indicating promising
locations for placing sensors. If the inversion is based on a
repeatable experiment, weights can also be used to specify the number
of experiments performed at each sensor location to control the
observation noise; see, e.g., \cite{HaberHoreshTenorio08}.  In many
inverse problems, however, experiments cannot be repeated or the
mathematical model is not an exact representation of the physical
phenomenon underlying the observations, such that the observation
error cannot be controlled.  Thus, %
we prefer weight vectors containing zeros and ones only,
indicating absence or presence of sensors over the candidate sensor
grid. Unfortunately, solving optimization problems for vectors with
binary components is a difficult combinatorial problem. Hence, we
employ a relaxation of the problem and consider weights $w_i \in [0,
  1]$. In Section~\ref{sec:sparsity}, we devise a method of recovering
the desired 0--1 structure using sparsifying penalties
combined with a continuation procedure.

Next, we describe the process of introducing the design vector $\vec{w}$ into the 
Bayesian inverse problem.
Since the design guides the collection of data, the weight vector $\vec{w}$ enters the 
inference problem~\eqref{equ:bayes-abstract} through the data likelihood~\eqref{equ:likelihood}.
The $\vec{w}$-weighted data-likelihood is given by,
\begin{equation}\label{equ:weighted-likelihood}
\like(\obs | \ipar; \vec{w}) \propto \exp\Big\{-\half(\iFF\ipar - \vec{d})^T\WW^{1/2}\ncov^{-1}\WW^{1/2}(\iFF\ipar - \vec{d})\Big\},
\end{equation}
where $\WW \in \R^{q \times q}$ is a diagonal matrix with weights on its diagonal.
For time-dependent problems, with observations collected at sensor locations at discrete points in time, 
we have $\iFF:\hilb \to \R^q$ with $q = \Ns \Ntau$, where $\Ns$ and $\Ntau$ are the 
number of candidate sensors and observation times, respectively.
Here, $\WW$ is a block-diagonal matrix having $\Ntau$ blocks, 
where each block is an $\Ns \times \Ns$ diagonal matrix with $\vec{w}$ on its diagonal.
The posterior covariance operator is thus given by 
\begin{equation}\label{eq:C0}
\C_\text{post}(\vec{w}) = (\iFFadj \WW^{1/2}\ncov^{-1}\WW^{1/2} \iFF + \C_0^{-1})^{-1}.
\end{equation}
From this point on, we work with the discretization 
of the infinite-dimensional problem, where we follow the discretization strategy 
described in Section~\ref{sec:bayes-disc}. 
In particular, we have the following discretized posterior covariance,  
\begin{equation}\label{eq:C0disc}
\postcov(\vec{w}) = (\Fadj \WW^{1/2}\ncov^{-1}\WW^{1/2} \FF + \prcov^{-1})^{-1}.
\end{equation}
The mean $\dparpost$ of the discretized posterior measure coincides
with the maximum a posteriori probability (MAP) estimate, given as the
solution of the minimization problem,
\begin{equation*}
    \min_{\dpar \in \R^\Nm} 
    \frac12
    \displaystyle\ip{\ncov^{-1}\WW^{1/2}\Big(\vec{d} - \FF\vec{m}\Big)}{\WW^{1/2}(\vec{d} - \FF\vec{m})} 
             + \frac12 \mip{\prcov^{-1}\Big(\vec{m} - \vec{m}_0\Big)}{\vec{m}-\vec{m}_0}.
\end{equation*}
The inverse of the Hessian $\H(\vec{w})$ of the above functional coincides with
$\postcov(\vec{w})$ defined in \eqref{eq:C0disc}. In the following, we
call the Hessian of the misfit term,
\[
   \HM(\vec{w}) = \Fadj\WW^{1/2}\ncov^{-1}\WW^{1/2}\F,
\] 
the \emph{misfit Hessian};
thus, $\H(\vec{w}) = \HM(\vec{w}) + \prcov^{-1}$.
In this paper, we consider
the case of independent observations, and hence $\ncov$ is a diagonal
matrix.  This assumption turns $\WW^{1/2}\ncov^{-1}\WW^{1/2}$ into a
diagonal matrix with diagonal entries $w_i/\sigma_i^2$.  For
simplicity of the presentation, we further assume that $\ncov$ is a
constant multiple of identity, i.e., $\ncov = \sigma \matrix{I}$, and
set $\sigma = 1$. Thus, the misfit Hessian takes the form $\HM =
\Fadj\WW\FF$. The algorithms presented below can be easily 
modified to accommodate a general diagonal noise covariance matrices $\ncov$.

\subsection{The OED problem} \label{sec:oed-formulation-discrete}
Finally, we can formulate the Bayesian A-optimal experimental design
problem as optimization problem for the weight vector $\vec w$.
Following the discussions in the previous sections, the OED objective
function is the trace
of the posterior covariance operator~\eqref{eq:C0}, which, after
discretization is given by
\eqref{eq:C0disc}. Additionally, we use a penalization to control the
sparsity of the design. Hence, the optimal design vector is the solution to the following 
optimization problem:
\begin{equation}\label{equ:oed-generic}
     \begin{aligned}
         \min_{\vec{w} \in \R^\Ns} &\quad \trace\big[\postcov(\vec{w})\big] + \upgamma \Phi(\vec{w}),\\
         \mbox{subject to}         &\quad 0 \leq w_i \leq 1, \quad i = 1, \ldots, \Ns,
     \end{aligned}
\end{equation}
where $\Phi:\R^\Ns_+ \to [0, \infty)$ is a penalty function and
  $\gamma\ge 0$ controls the sparsity of the design.  We note that the
  function $\vec{w} \mapsto \trace\big[\postcov(\vec{w})\big] =
  \trace\big[\H(\vec{w})^{-1}\big]$ is strictly convex due to strict
  convexity of $\matrix{X}\mapsto \trace(\matrix{X}^{-1})$ on the cone
  of symmetric positive definite matrices (see~\cite[p.~82]{Pazman86})
  and the fact that $\vec{w}$ enters linearly in
  $\H(\vec{w})$. Therefore, if the penalty function $\Phi$ is convex,
  \eqref{equ:oed-generic} has a unique solution.  An example of a
  convex penalty is given by $\Phi(\vec{w}) = \vec{1}^T \vec{w}$,
  i.e., an $\ell_1$-penalty, whose sparsening property has been used
  extensively in compressive
  sensing~\cite{Donoho06a,Cand`esRombergTao06} and has also been
  adapted to OED for inverse
  problems~\cite{HaberHoreshTenorio08,HaberMagnantLuceroEtAl12}.  In
  this paper, we will use a continuation approach involving a family
  $\{ \Phi_\eps \}_{\eps > 0}$ of penalty functions, which approximate
  the $\ell_0$-``norm''. This allows us to find binary optimal design vectors.

We remark that there exists an alternative interpretation of the
A-optimal design criterion in the Gaussian linear case considered
here. Namely, minimizing the trace of the posterior covariance is
equivalent to minimizing the average mean square error (MSE) of the
posterior mean, where the average is with respect to the prior
measure. This average MSE is also referred to as the Bayes risk of the posterior mean.
MSE is a concept in frequentist inference, in which the
posterior mean is interpreted as an estimator for the unknown
parameter. This frequentist point of view of A-optimal design is used
in~\cite{HaberHoreshTenorio08,HaberMagnantLuceroEtAl12}.
For completeness of our presentation, we detail this relation
between average MSE and the trace of posterior covariance
in Appendix~\ref{appdx:MSE}.

\section{Numerical solution of the OED problem~\eqref{equ:oed-generic}}\label{sec:oed-algorithms}
We begin this section with deriving a decomposition of
the misfit Hessian in terms of contributions from different sensors.
Then, in Section~\ref{sec:aopt-analytic}, 
we present an approximation of $\trace(\postcov)$ using a randomized trace estimator, 
and derive expressions for the gradient of the resulting OED objective function.
Subsequently, in Sections~\ref{sec:lowrank-approx} and~\ref{sec:aopt-alg},
we present algorithmic components that allow efficient realization of these computations.
Finally, in Section~\ref{sec:sparsity}, we
discuss methods to control the sparsity of the design, i.e., the
number of allocated sensors.

\subsection{Decomposition of the misfit Hessian}
\label{sec:forward_map}

The misfit Hessian plays an important role in the derivative computation
of the OED objective function with respect to the sensor
location weights.
Therefore, we first
derive a decomposition of~$\HM$ as 
a weighted sum of terms corresponding to individual sensor locations.

We consider a linear parameter-to-observable map $\FF$, which
involves a time-dependent PDE. Parameter-to-observable maps with stationary equations
are included as special case in the discussion below by considering a
single time step.
We consider observations at the candidate sensor locations $\vec{x}_1,
\ldots, \vec{x}_\Ns$ in $\D$. For each sensor location, the time
evolution of the observation is discretized using Lagrange elements in
time (we use piecewise linear elements in this paper) for the nodal
time instances $\tau_1, \ldots, \tau_\Ntau$ ($\tau_i \in [0, T]$
for $i=1,\ldots,\Ntau$). These observation times are independent of
the time steps used in the integration of the PDE.

The parameter-to-observable map $\FF$ takes a parameter vector
$\vec{m} \in \R^\Nm$ and maps it to the space-time observation vector
$\dd \in \R^{\Ns\Ntau}$:
\begin{equation}\label{equ:param-to-obs}
\FF:\vec{m} 
   \,
   \stackrel{\sol}{\longmapsto}  
   \,
   \uu 
   \,
   \stackrel{\BB}{\longmapsto} 
   \,
   \dd. %
\end{equation}
Here,  $\sol$ is the discretized PDE solution operator, $\uu \in \R^{n(\Nt+1)}$ is the space-time 
solution vector, and $\BB$ is the space-time observation operator.
We target designs where the sensor locations coincide for all time
observations. Thus, the sensor weight matrix
$\WW$ can be written as $\WW = \sum_{j = 1}^\Ns w_j \matrix{E}_j$, 
where $\matrix{E}_j$ is an $\Ns\Ntau \times \Ns\Ntau$ block-diagonal matrix, with 
$\Ntau$ blocks, with each block equal to $\vec{e}_j \otimes \vec{e}_j
= \vec{e}_j\vec{e}_j^T$; here $\vec{e}_j$ denotes the $j^{\text{th}}$ coordinate vector in $\R^\Ns$.
Thus, the misfit Hessian for a weight vector $\vec w$ can be decomposed as
$
\HM(\vec{w}) = \Fadj \WW \F = \sum_{j = 1}^\Ns w_j  \Fadj \matrix{E}_j \FF,
$
where the matrices $\Fadj \matrix{E}_j \FF$, $j = 1, \ldots, \Ns$
are the \emph{atoms} corresponding to the different sensor locations.
This decomposition, which is also used in~\cite{Ucinski05}, reveals the identity
\begin{equation} \label{equ:hessian-partial}
\delj{\HM(\vec w)} = \Fadj \matrix{E}_j \FF.
\end{equation}
Next, we approximate the OED objective function using
trace estimators.

\subsection{The OED objective function and its derivative}\label{sec:aopt-analytic}
We consider the problem~\eqref{equ:oed-generic} and recall that
$\postcov(\vec{w}) = \H(\vec{w})^{-1}$, where the Hessian is an
$\M$-symmetric linear mapping on $\R^\Nm_M$.  For the numerical
solution of the OED problem, we approximate the trace of $\H(\vec
w)^{-1}$ using a randomized trace estimator (see
Section~\ref{sec:rand}). 
The trace estimator-based OED objective functional is
\begin{equation}\label{equ:aopt_obj}
 \Theta(\vec{w}) := 
                       \frac{1}{\Ntr} \sum_{i = 1}^\Ntr \mip{\vec{z}^{(i)}}{\matrix{H}(\vec w)^{-1}\vec{z}^{(i)}},  
\end{equation}
where $\vec{z}^{(i)} = \M^{-1/2} \vec{y}^{(i)}$, $i = 1, \ldots, \Ntr$, with
$\vec{y}^{(i)}$ appropriately chosen random vectors\footnote{In the
  present work, we rely on Gaussian trace estimators. See
  Appendix~\ref{appdx:trace} for a justification of trace estimation in the context of weighted inner products.}.
Therefore, we consider the following OED optimization problem in our numerical computations,
\begin{equation}\label{equ:oed-est}
     \begin{aligned}
         \min_{\vec{w} \in \R^\Ns} & \quad \Theta(\vec{w}) + \upgamma
       \Phi(\vec{w}), \\
         \mbox{subject to}        & \quad 0 \leq w_i \leq 1, \quad i = 1, \ldots, \Ns.
     \end{aligned}
\end{equation}
Since we will use gradient-based optimization methods to
solve~\eqref{equ:oed-est}, we need to compute the gradient of
$\Theta(\vec w)$ with respect to $\vec w$. For $j\in
\{1,\ldots,\Ns\}$, we obtain
\begin{equation*}
\begin{aligned}
   \delj{\Theta(\vec w)} &= \frac1{\Ntr} \sum_{i = 1}^\Ntr \mip{\vec{z}^{(i)}}{\delj{\matrix{H}(\vec w)^{-1}}\vec{z}^{(i)}}
                  \\&= -\frac{1}{\Ntr} 
                     \sum_{i = 1}^\Ntr \mip{\vec{z}^{(i)}}{\H(\vec w)^{-1}\delj{\matrix{H}(\vec w)}\H(\vec w)^{-1}\vec{z}^{(i)}}.
\end{aligned}
\end{equation*}
The $\M$-symmetry of $\H(\vec w)^{-1}$, the fact that the prior does not
depend on $\vec w$, and denoting $\vec{q}^{(i)} =
\H^{-1}\vec{z}^{(i)}$
yields
\begin{equation}\label{equ:varphi-deriv}
   \delj{\Theta(\vec{w})} =  -\frac{1}{\Ntr} \sum_{i = 1}^\Ntr \mip{\vec{q}^{(i)}}{\delj{\HM(\vec w)}\vec{q}^{(i)}}, \quad
   j = 1, \ldots, \Ns.
\end{equation}
Using~\eqref{equ:hessian-partial} and denoting $\dd^{(i)} = (\vec
d^{(i),1} , \ldots, \vec d^{(i),\Ntau})^T := \FF\vec{q}^{(i)}$, with
$\vec{d}^{(i)} \in \R^\Ns$ the spatial observations corresponding
to time $\tau_i$, each of the summands in \eqref{equ:varphi-deriv} can
be computed as
\begin{equation} \label{equ:summand}
\begin{aligned}
    \mip{\vec{q}^{(i)}}{\delj{\HM(\vec w)}\vec{q}^{(i)}}%
                                         &= \ip{\FF\vec{q}^{(i)}}{\matrix{E}_j\FF \vec{q}^{(i)}} \\
                                         &= \sum_{\ell = 1}^\Ntau [\vec{d}^{(i),\ell}]^T 
                                           (\vec{e}_j\otimes \vec{e}_j) [\vec{d}^{(i),\ell}]
                                         = \sum_{\ell = 1}^\Ntau \big[ d^{(i),\ell}_j \big]^2. %
\end{aligned}
\end{equation}
To summarize, the evaluation of $\Theta(\vec w)$ requires $\Ntr$
multiplications of vectors with $\matrix H(\vec w)^{-1}$. The
computation of the derivative of $\Theta(\vec w)$ with respect to
$\vec w$ additionally requires $\Ntr$ evaluations of the
parameter-to-observable map $\FF$. 
In the Sections~\ref{sec:lowrank-approx} and~\ref{sec:aopt-alg} below, we discuss
the efficient realization of these computations.

\subsection{Low-rank approximation of the prior-preconditioned parameter-to-observable map}
\label{sec:lowrank-approx}
As shown above, the repeated application of $\H(\vec w)^{-1}$ is
necessary to compute $\Theta (\vec w)$ and the gradient of
$\Theta(\vec w)$ with respect to $\vec w$.  In a numerical
optimization algorithm to solve \eqref{equ:oed-est}, these
computations are required in each iteration. Despite the use of a
trace estimator, this is computationally demanding and can render OED
for large-scale Bayesian inverse problems infeasible.

As a remedy, we construct a surrogate model for the
parameter-to-observable map that can be used to efficiently compute
the application of $\matrix H(\vec w)^{-1}$ to vectors. The surrogate
construction exploits the fact that for a large class of
infinite-dimensional inverse problems, $\FF$ can be well approximated
by a low-rank operator due to properties of the underlying PDE and the
limited number of observations; see for instance
\cite{FlathWilcoxAkcelikEtAl11,Bui-ThanhGhattasMartinEtAl13}. Thus, we
can compute a low-rank SVD surrogate for $\FF$ \emph{upfront}, and use
this surrogate in all subsequent computations involving $\FF$ to find
an optimal experimental design.

Since only parameters consistent with the data {\em and} the prior
have a significant posterior probability, only those parameters can
influence the OED. Thus, it suffices to compute a surrogate
of the \emph{prior-preconditioned} parameter-to-observable map, $\FT :=
\FF \prcov^{1/2}$.  The smoothing property of the priors usually
employed in infinite-dimensional Bayesian inversion\footnote{Note that
  a smoothing property is necessary to render the infinite-dimensional
  Bayesian inverse problem well-posed, \cite{Stuart10}.} results in
$\FT$ having faster decaying singular values than $\FF$.  Thus, the
number of forward and adjoint PDE solves required to construct an
accurate low-rank SVD surrogate for $\FT$ is usually smaller than that for
$\FF$.  For the model problem in Section~\ref{subsec:lowrankF}, this
results in a significant speedup for the construction of the low-rank
surrogate.

We employ randomized SVD (see Section~\ref{sec:rand}) to
compute the low-rank surrogate of $\FT$. This algorithm only requires
the applications of $\FT$ and $\FT^*$ to a set of independent random
vectors, which is convenient and can often be implemented efficiently
in large-scale computations.  In the next section, we show how the
low-rank surrogate of $\FT$ can be used to compute $\Theta(\vec
w)$ and its gradient with respect to $\vec w$.

\subsection{Efficient computation of $\Theta (\vec w)$ and its derivatives}
\label{sec:aopt-alg}
First, we
discuss the efficient computation of $ \H(\vec
w)^{-1}\vec{z}$ for $\vec{z} \in \R^\Nm$, as required to evaluate
$\Theta(\vec w)$ and its derivative.
Note that the inverse of the
Hessian $\H(\vec w)$ can be written
as~\cite{Bui-ThanhGhattasMartinEtAl13, MartinWilcoxBursteddeEtAl12},
\begin{equation}\label{eq:Homegainv}
    \H(\vec w)^{-1} =
              \prcov^{1/2}\big( \HT(\vec w)  + \matrix{I})^{-1}\prcov^{1/2},
\end{equation}
where, for a weight vector $\vec w\in \mathbb R^\Ns$, $\HT(\vec w) =
\FT^* \WW \FT$ is the \emph{prior-preconditioned} misfit Hessian.
Assuming a rank-$r$ surrogate $\FT_r$ for $\FT$ is available,
$\HTr(\vec w) = \FT^*_r \WW \FT_r$ is the resulting approximation of
$\HT(\vec w)$ used below. Note that for every weight vector $\vec w$,
the rank of $\HTr(\vec w)$
is less than or equal to $r$. One approach to apply $(\HTr +
\matrix{I})^{-1}$ to a vector is to first compute the spectral decomposition
$\HTr(\vec w) =
\matrix{V}_r \matrix{\Lambda}_r \matrix{V}_r^*$, where $\matrix{V}_r$ is
the matrix containing the eigenvectors of $\HTr(\vec w)$ as its columns, and
$\matrix{\Lambda}_r$ is the diagonal matrix with $r$ largest eigenvalues on the
diagonal.\footnote{Note that thanks to the low-rank surrogate of
  the prior-preconditioned parameter-to-observable map $\FT_r$, this
  computation does not require the solution of PDEs.}  Then, by
the Sherman-Morrison-Woodbury formula~\cite{GolubVan96},
we have $(\matrix{V}_r
\matrix{\Lambda}_r\matrix{V}_r^* + \matrix{I})^{-1} = \matrix{I} -
\matrix{V}_r\matrix{D}_r\matrix{V}_r^*$, where $\matrix{D}_r$ is a
diagonal matrix with $\lambda_i / (1 + \lambda_i)$, $i = 1, \ldots,
r$, on its diagonal.\footnote{Note that there exist alternatives to
  computing a spectral decomposition of $\HTr(\vec w)$ and then using the
  Sherman-Morrison-Woodbury formula. For instance, one can use a
  Krylov method to compute $\H(\vec w)^{-1}\vec{z}$.} Defining 
\begin{equation}\label{equ:qhat-SMW}
 \vec{\hat{q}} =
 \big(\matrix I - \Vr \Dr \Vr^*\big) \prcov^{1/2} \vec{z}
\end{equation}
allows us to compute
\begin{equation} \label{equ:compute-q}
\vec q = \H(\vec w)^{-1}\vec z \approx  \prcov^{1/2}\vec{\hat{q}},
\end{equation}
where the approximation is only due to the use of $\FT_r$ instead of $\FT$.
To compute $\FF \vec{q}$, as needed in~\eqref{equ:summand}, we use 
\begin{equation}\label{equ:compute-Fq}
\FF \vec{q} = \FF \prcov^{1/2} (\HT(\vec w) + \matrix{I})^{-1} \prcov^{1/2} \vec{z} \approx \FT_r \vec{\hat{q}},
\end{equation}
where the approximation is again due to the use of the low-rank
surrogate for $\FT$. 
We now summarize the procedure for computing $\Theta(\vec w)$
and its gradient. Note that once %
$\FT_r$ is available, our method does not require further forward or
adjoint PDE solves. 
\begin{algorithm}
\caption{Algorithm for computing $\Theta(\vec w)$ and its gradient $\vec{g} = \Grad\Theta(\vec w)$ 
with random vectors $\vec{z}^{(i)}$, $i = 1, \ldots, \Ntr$ for the trace estimator.}
\begin{algorithmic}
\STATE Compute $\HTr(\vec{w}) = \matrix{V}_r \matrix{\Lambda}_r \matrix{V}_r^*$
       \hfill \COMMENT{low-rank approx to $\HT(\vec{w}) = \FT^* \WW \FT$}
\STATE Set $\matrix{D}_r = \diag\Big( \frac{\lambda_1}{1 + \lambda_1}, \ldots, \frac{\lambda_r}{1 + \lambda_r} \Big)$
\FOR{$i = 1$ to $\Ntr$}
   \STATE Compute
       $\vec{\hat{q}}^{(i)} = \Big( \matrix{I} - \matrix{V}_r \matrix{D}_r \matrix{V}_r^*\Big) \prcov^{1/2} \vec{z}^{(i)}$
       \hfill \COMMENT{as in~\eqref{equ:qhat-SMW}}
   \STATE Compute
       $\vec{q}^{(i)} = \prcov^{1/2} \vec{\hat{q}}^{(i)}$
       \hfill \COMMENT{compute $\vec{q}^{(i)} = \matrix{H}^{-1}\vec{z}^{(i)}$ as in~\eqref{equ:compute-q}}
\ENDFOR
\STATE Compute $\displaystyle \Theta(\vec{w}) = \frac{1}{\Ntr} \sum_{i = 1}^\Ntr \mip{\vec{z}^{(i)}}{\vec{q}^{(i)}}$
\hfill \COMMENT{evaluation of $\Theta(\vec{w})$}
\STATE $\vec{g} = \vec{0}$ 
\hfill \COMMENT{initialize the gradient vector}
\FOR{$i = 1$ to $\Ntr$}
   \STATE $\dd = \FT_r \vec{\hat{q}}^{(i)}$
   \hfill \COMMENT{compute $\dd = \FF \vec{q}^{(i)}$ as in~\eqref{equ:compute-Fq}}
   \FOR{$j = 1$ to $\Ns$}
      \STATE $\displaystyle g_{i,j} = \sum_{\ell = 1}^{\Ntau} \left( d_j^\ell\right)^2$ 
       \hfill \COMMENT{$\displaystyle g_{i,j} = \mip{\vec{q}^{(i)}}{\delj{\HM}\vec{q}^{(i)}}$ as in~\eqref{equ:summand}}
   \ENDFOR
   \STATE $\vec{g} = \vec{g} - \frac{1}{\Ntr}\vec{g}_i$
\ENDFOR
\STATE 
\end{algorithmic}
\label{alg:aopt}
\end{algorithm}

We close this subsection with two remarks concerning the computations
involved in Algorithm~\ref{alg:aopt}.
First, note that while the
low-rank SVD surrogate for $\FT$ relieves the OED method of repeated
PDE solves, we still require the repeated applications of
$\prcov^{1/2}$.  Due to our choice of the prior as the inverse of a
squared elliptic operator, this amounts to an elliptic solve, for
which optimal complexity solvers (e.g., multigrid) are available.

Next, note that we require the application of $\M^{-1/2}$ to random
vectors for the trace estimator in the OED method; see
\eqref{equ:aopt_obj}.  The same operation, which transforms vectors
from the Euclidean space $\R^n$ to $\R^n_\M$ (see
Section~\ref{sec:bayes-disc}) via a mapping that preserves the
inner products, is also needed to draw samples from the posterior
distribution (see~\eqref{equ:sample}).  Since the explicit computation
of the matrix square root is expensive in high dimensions, we utilize
an iterative algorithm to compute the application of the inverse
matrix square root to vectors~\cite{ChenAnitescuSaad11}.  The method
relies on the approximation of the square root via orthogonal
polynomials on an interval containing the spectrum of the matrix;
since the convergence of this method is fastest for matrices with
clustered eigenvalues, we apply it to the mass matrix $\M$,
symmetrically preconditioned with the lumped mass matrix, $\ML$. The
eigenvalues of the resulting matrix, $\Mt = \ML^{-1/2}\M\ML^{-1/2}$,
are clustered around 1, resulting in fast convergence when applying
$\Mt^{-1/2}$ to vectors as discussed in~\cite{ChenAnitescuSaad11}.
Then, instead of $\M^{-1/2}$ we use the matrix $\matrix{L} =\ML^{-1/2}
\Mt^{-1/2}$ as the isomorphism between $\R^n$ and $\R^n_\M$.  Note, in
particular, that for all $\vec{x}$ and $\vec{y}$ in $\R^n$ %
\[
\mip{\L \vec{x}}{\L \vec{y}} 
= \vec{x}^T \Mt^{-\half}\ML^{-\half} \M \ML^{-\half} \Mt^{-\half} \vec{y} %
= \vec{x}^T \Mt^{-\half} \Mt \Mt^{-\half} \vec{y} = \ip{\vec{x}}{\vec{y}},
\]
which shows that $\L$ is a Hilbert space isomorphism between $\R^n$
and $\R^n_\M$.

\subsection{Sparsity-enforcing penalty functions}
\label{sec:sparsity}
Here, we discuss the penalty term $\upgamma\Phi(\cdot)$
in~\eqref{equ:oed-generic}.  If $\Phi(\cdot)$ favors sparse solutions,
the degree of sparsity depends on the choice of $\upgamma$. However,
$\upgamma$ only provides an indirect control of the sparsity of the
design and in practice adjusting $\upgamma$ and resolving
\eqref{equ:oed-generic} might be necessary to obtain a design with a
specified (or close to a specified) number of sensors.

Next, we discuss two possible choices for the penalty function $\Phi(\cdot)$.
An $\ell_1$-penalty
function $\Phi(\vec{w}) = \vec{1}^T \vec{w}$ (see also
\cite{HaberHoreshTenorio08,HaberHoreshTenorio10,HaberMagnantLuceroEtAl12}),
is known to result in sparse solution vectors $\vec w$. While this choice is convenient due to the
convexity of $\Phi$, it results in a weight vector $\vec w$ whose
components can take on any value in $[0,1]$. Unfortunately, the direct
interpretation of weights $w_i\not\in\{0,1\}$ for the placement of
sensors is unclear.  A practical approach to use such a solution
vector is to place sensors in locations where the corresponding
weight does not vanish.
An alternative to
$\ell_1$-sparsification is to use non-convex functions for $\Phi$ that
directly lead to binary weight vectors. This approach is
discussed next.

To obtain binary weight vectors, we select penalties that
successively approximate the $\ell_0$-``norm'',\footnote{Note that
  $\norm{\vec{x}}_0$ is in fact not a \emph{norm}.} which, for
$\vec{x} \in \R^n$, $\norm{\vec{x}}_0$ is defined as the number of
non-zero elements in $\vec{x}$.  One choice for such a family of
penalizations is $\Phi_q(\vec{w}) = \norm{\vec{w}}_q^q$, for $q <
1$. Since these penalization functions are not Lipschitz continuous at the origin,
we use an alternative family of regularizations, namely, for $\eps>0$,
\begin{equation}\label{eq:feps}
 \Phi_\eps(\vec{w}) := \sum_{i = 1}^\Ns f_\eps(w_i),\quad \text{ where }\:\:
   f_\eps(w) = \begin{cases}\displaystyle 
               \frac{w}{\eps}, \quad &0 \le w \le \frac12 \eps,\\
               p_\eps(w), \quad &\frac12\eps < w \le 2\eps, \\
               1,   \quad &2\eps < w \le 1.
               \end{cases}
\end{equation}
Here, $p_\eps(\cdot)$ is the uniquely defined third order polynomial
that makes $f_\eps(x):[0, 1] \to [0, 1]$ continuously differentiable. 
In Figure~\ref{fig:penfun}, we show these penalty functions for
different values of $\eps$; note how $f_\eps(\cdot)$ approximates the
$\ell_0$-``norm'' as $\eps\to 0$.
\tikzset{dashdot/.style={dash pattern=on .4pt off 3pt on 4pt off 3pt}}
\begin{figure}[ht]\centering
\begin{tikzpicture}[]
\begin{axis}[compat=newest, width=4.5cm, height=3.8cm, scale only axis,
    xlabel = $x$, ylabel=$f_\eps(x)$, xmin=0, xmax=1,
    ymin=0, ymax=1.1,
    legend style={font=\small,nodes=right}, legend pos= south east]
\addplot [color=black, thick] table[x=x,y=y] {penfun1.txt};
\addlegendentry{$\eps = 1/2$}
\addplot [color=black, thick, dashed] table[x=x,y=y] {penfun2.txt};
\addlegendentry{$\eps = 1/4$}
\addplot [color=black, thick, dotted] table[x=x,y=y] {penfun3.txt};
\addlegendentry{$\eps = 1/8$}
\addplot [color=black, thick, dashdot] table[x=x,y=y] {penfun4.txt};
\addlegendentry{$\eps = 1/16$}
\end{axis}
\end{tikzpicture}\vspace{-2ex}
\caption{Graphs of $f_\eps(\cdot)$ as defined in \eqref{eq:feps} for $\eps = 1/2^i$, $i = 1, \ldots, 4$.}
\label{fig:penfun}
\end{figure}

To cope with potential multiple local minima due to the non-convexity
of the penalties $\Phi_\eps$, we use a continuation procedure
with respect to $\eps$. We fix the regularization parameter $\upgamma$
in~\eqref{equ:oed-generic} and solve the OED problem with
$\ell_1$-penalty, resulting in a weight vector $\vec{w}^0$.  We then
choose a decreasing sequence of positive real numbers $\{\eps_i\}_{i
  \geq 1}$, and compute a new weight vector $\vec{w}^i$ by
solving~\eqref{equ:oed-generic} with penalty function
$\Phi_{\eps_i}(\vec{w})$ using $\vec{w}^{i-1}$ as initialization.  In
practice, once $\eps_i$ is sufficiently small, the solution
$\vec{w}^i$ is a 0--1 vector that remains unchanged as $\eps$ is
further decreased.

A similar continuation strategy is used in
topology optimization~\cite{Bendsoe95,BendsoeSigmund03}, where the
target is to design an ``optimal structure'' by deciding on an optimal
distribution of material in a physical domain $\D$.  In this
application, one seeks a \emph{density} function $\rho:\D \to \{0,
1\}$ that characterizes the absence or presence of material in $\D$.
One of the main approaches
to solving topology optimization problems is to relax the requirement
of $\rho \in \{0, 1\}$ to $0 \leq \rho \leq 1$ and to solve a sequence
of optimization problems with successively steeper penalty functions
to approach a 0--1 solution $\rho$---similar our $\Phi_\eps$-continuation approach outlined above.

In the rest of this paper, we refer to the designs obtained via the above 
continuation method as \emph{$\Phi_\eps$-sparsified designs}. 
Optimal designs obtained by solving~\eqref{equ:oed-generic} 
with $\ell_1$-penalty (followed by thresholding to obtain a 0--1 weight 
vector) are referred to as \emph{$\ell_1$-sparsified designs}.  
In the numerical experiments presented in Section~\ref{sec:numerics}, 
we compare $\ell_1$- and $\Phi_\eps$-sparsified designs.
\section{Model problem setup}\label{sec:model}
Here, we present the model problem used
to study the methods for OED presented in this paper.  We consider a
time-dependent advection-diffusion equation in which we seek to infer
an unknown initial condition from point measurements; see also
\cite{AkcelikBirosDraganescuEtAl05,FlathWilcoxAkcelikEtAl11,PetraStadler11}.

\subsection{The parameter-to-observable map}
The PDE in the parameter-to-observable map models diffusive transport
in a domain $\D \subset \R^d$, which is depicted in
Figure~\ref{fig:ad_dom}(a) for $d = 2$ and~\ref{fig:ad_dom}(c) for $d
= 3$.  The domain boundaries $\del \D$ include the outer edges/faces
as well as the internal boundaries of the rectangles/hexahedra, which
model buildings.
The parameter-to-observable map maps an initial condition $\ipar\in
L^2(\D)$ to spatial and temporal point observations through the
solution of the advection-diffusion equation $u(\vec{x}, t)$ given by
\begin{equation}\label{eq:ad}
  \begin{aligned}
    u_t - \kappa\Delta u + \mathbf{v}\cdot\nabla u &= 0 & \quad&\text{in
    }\D\times (0,T), \\
    u(\cdot, 0) &= \ipar  &&\text{in } \D , \\
    \kappa\nabla u\cdot \vec{n} &= 0 &&\text{on } \partial\D \times (0,T).
  \end{aligned}
\end{equation}
Here, $\kappa > 0$ is the diffusion coefficient and $T > 0$ is the final
time. In
our numerical experiments, we use $\kappa = 0.001$ for the two-dimensional
problem and $\kappa = 0.003$ for the three-dimensional problem.  The velocity field
$\vec{v}$, shown in Figures~\ref{fig:ad_dom}(b), (c) for the two- and the
three-dimensional problems, respectively, is computed by solving the following steady-state
Navier-Stokes equation with the side walls driving the flow:
\begin{equation}\label{eq:fp:navst}
  \begin{aligned}
    - \frac{1}{\operatorname{Re}} \Delta \vec{v} + \nabla q + \vec{v} \cdot \nabla
    \vec{v} &= 0 &\quad&\text{ in
    }\D,\\
    \nabla \cdot \vec{v} &= 0 &&\text{ in }\D,\\
    \vec{v} &= \vec{g} &&\text{ on } \partial\D.
\end{aligned}
\end{equation}
Here, $q$ is pressure, $\text{Re}$ is the Reynolds number,
set to $\text{Re} = 50$ in the present examples. The Dirichlet boundary data
$\vec{g} \in \R^d$ is given by 
$\vec{g} = \vec{e}_2$ on the left wall of the domain, 
$\vec{g}=-\vec{e}_2$ on the right wall,  and $\vec{g} = \vec{0}$ everywhere else
(see, e.g., Figure~\ref{fig:ad_dom}(a) for $d = 2$). 
\begin{figure}\centering
\begin{tabular}{ccc}
\includegraphics[width=0.30\textwidth]{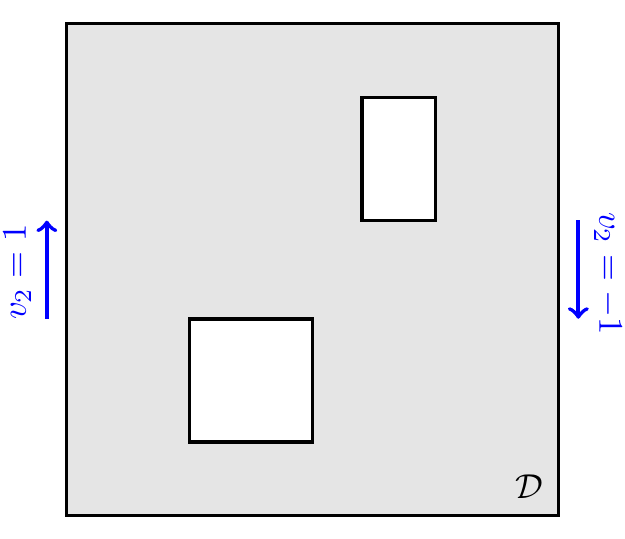}&
\includegraphics[width=.25\textwidth]{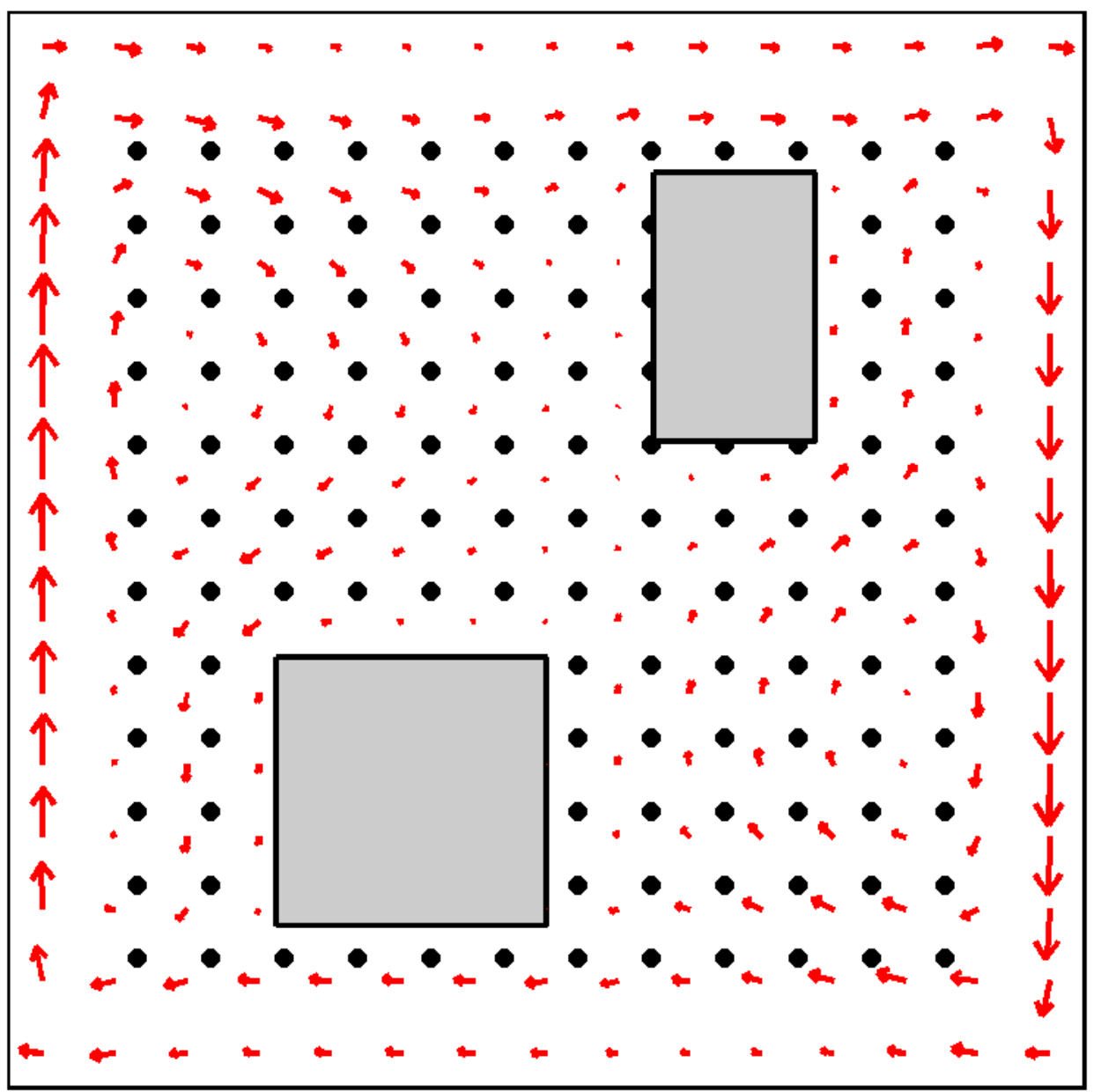} &
\includegraphics[width=.29\textwidth]{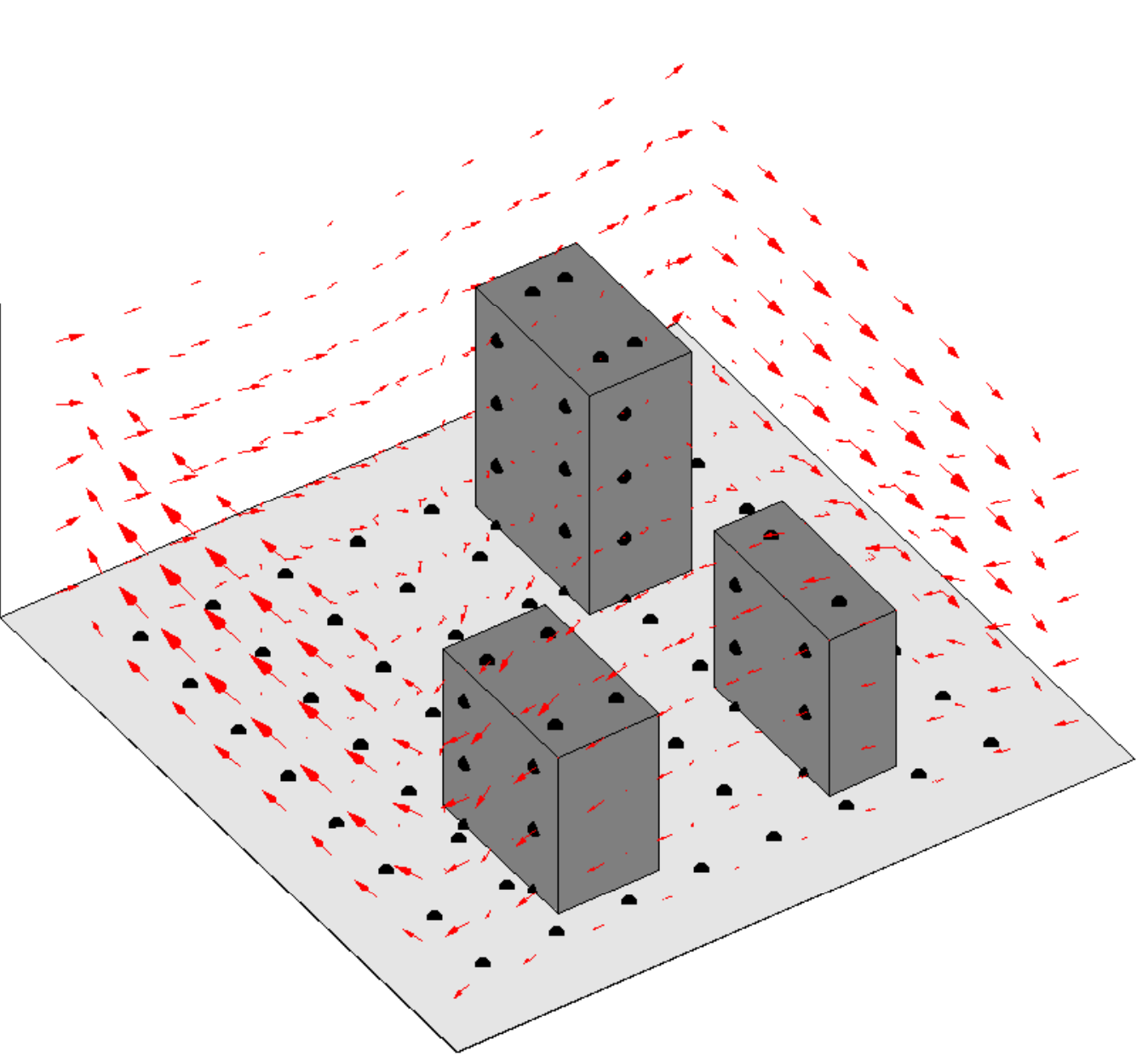}\\
(a) & (b) & (c)
\end{tabular}
\caption{(a) The computational domain $\D$ for the two-dimensional model problem is $[0,1]^2$ with two rectangular regions (representing buildings) removed.
(b) The velocity field $\vv$ (red arrows) and the candidate sensor locations (black dots) for the two-dimensional problem.
(c) The computational domain $\D$ for the three-dimensional model problem is $[0,1]^2 \times [0, 0.5]$ with 
the buildings (gray blocks) removed; shown are also the velocity field (red arrows) and the 
candidate sensor locations (black dots).}
\label{fig:ad_dom}
\end{figure}
To illustrate the physics of the forward problem for a given initial
condition, we show snapshots of the time evolution of the state
variable $u$ in Figure~\ref{fig:state3d}; for snapshots of the time
evolution for the two-dimensional model problem we refer to~\cite{PetraStadler11}.
\begin{figure}[ht]\centering
\includegraphics[width=.33\textwidth]{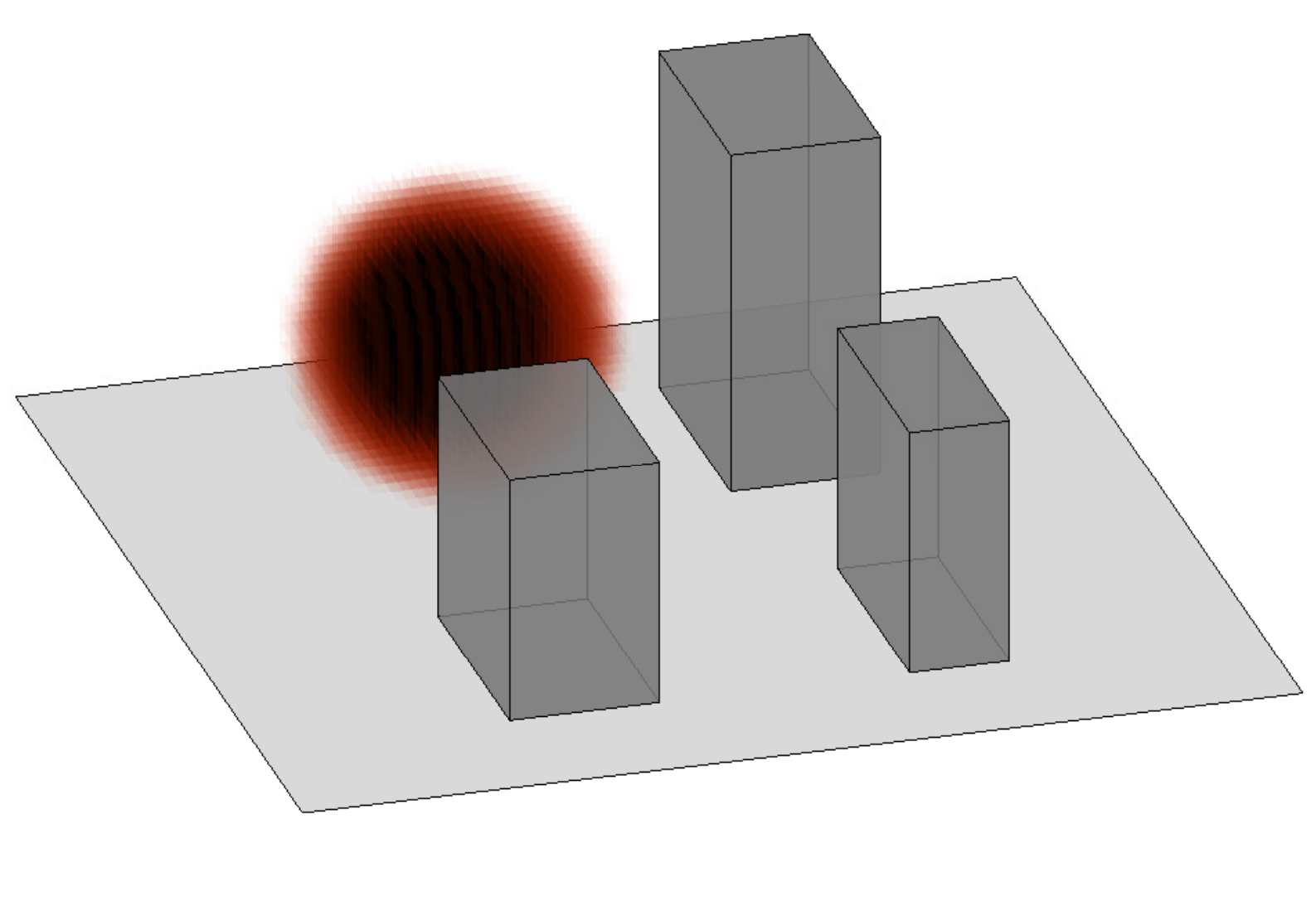}\hfill%
\includegraphics[width=.33\textwidth]{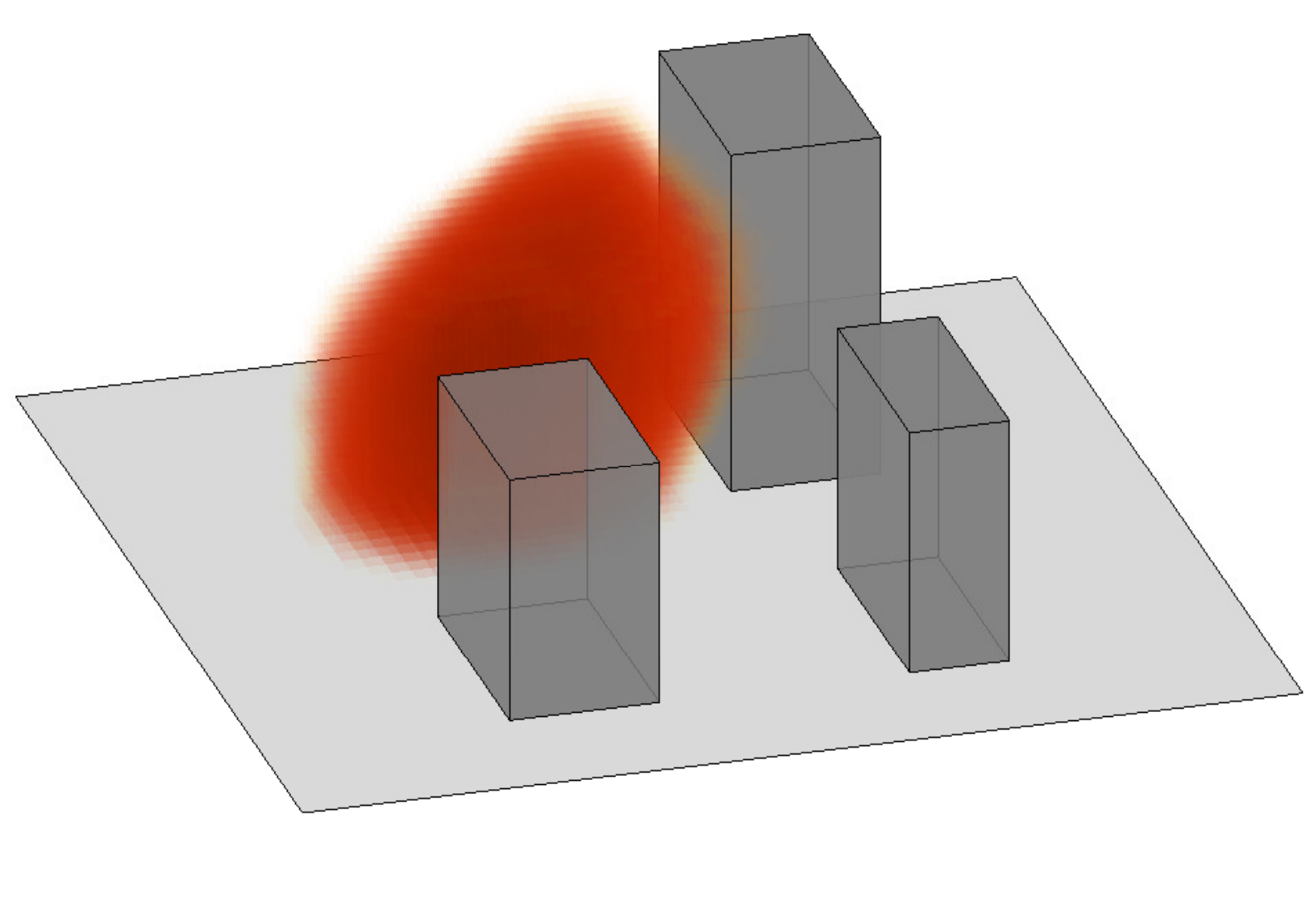}\hfill%
\includegraphics[width=.33\textwidth]{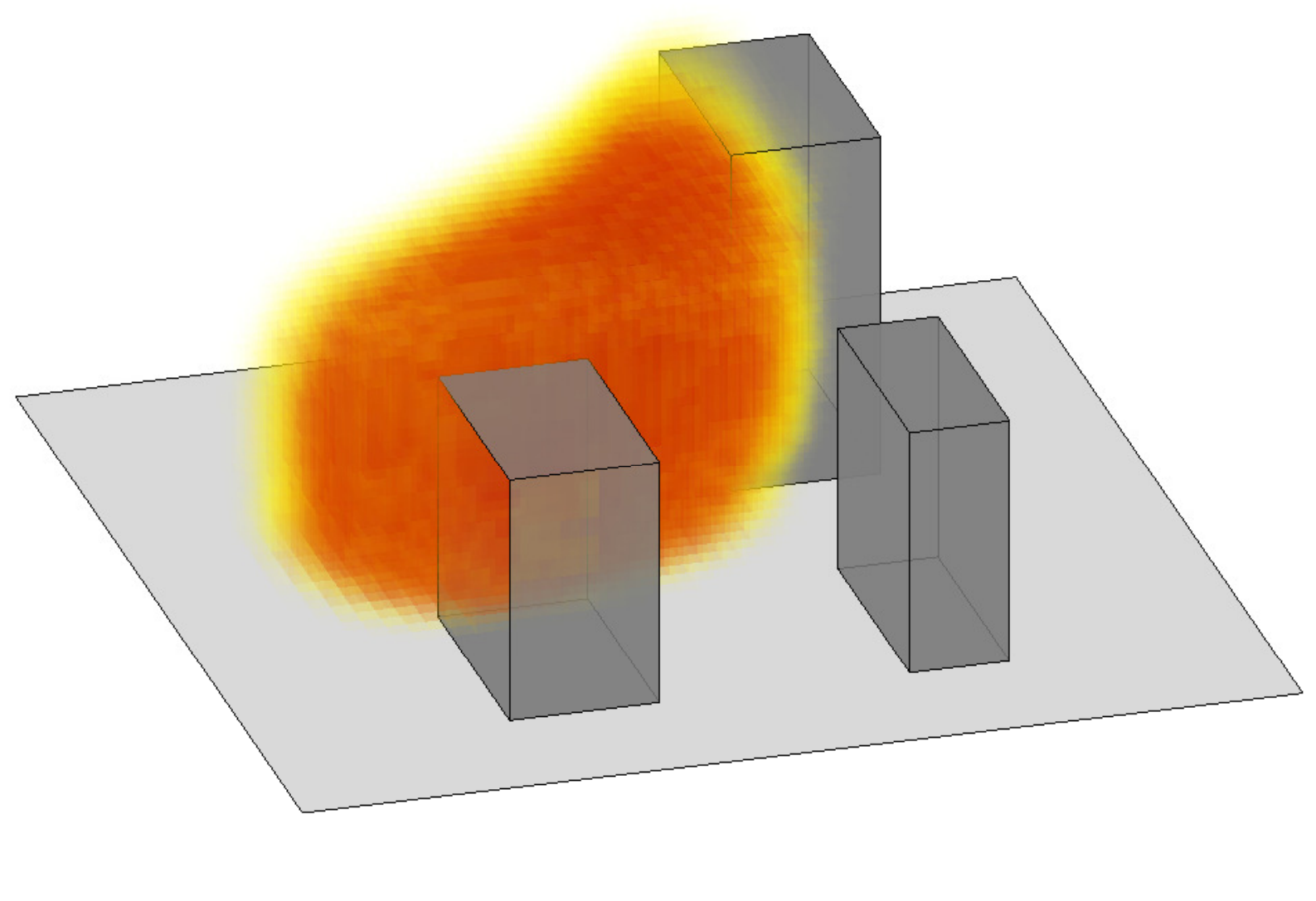}
\caption{Volume renderings of the time evolution of $u$ for the three-dimensional
  model problem. The initial condition $\ipar$ is shown in the left,
  and the middle and right images correspond to snapshots at times $t
  = 1$ and $t=2$.}
\label{fig:state3d}
\end{figure}

The observation operator $\mathcal{B}$ then extracts the values of $u$ on a set of
sensor locations $\{\vec{x}_1, \ldots, \vec{x}_\Ns\} \subset \D$ at
times $\{\tau_1, \ldots, \tau_\Ntau\} \subset [0, T]$. To summarize,
the infinite-dimensional linear parameter-to-observable map $\iFF$
maps the initial condition $\ipar$ to $q=\Ns\Ntau$ observations as
follows: First, we solve the time-dependent advection-diffusion
equation \eqref{eq:ad}, resulting in the space-time solution
$u=u(m)$. Then, the observation operator extracts point values from
$u$ at the measurement locations and times, i.e., evaluates $\mathcal{B}u$.
The corresponding discrete parameter-to-observable map $\FF$ is
obtained by discretizing $\ipar$ and $u$ using, for instance,
finite elements.

\subsection{The Bayesian inverse problem}

Following the Bayesian setup in Section~\ref{sec:bayes}, we utilize 
a Gaussian prior measure $\priorm = \GM{\ipar_0}{\C_0}$,
with $\C_0=\Acal^{-2}$ where $\Acal$ is an elliptic differential operator as 
described in Section~\ref{sec:bayes}, 
and use an additive
Gaussian noise model. Therefore, the solution of the Bayesian inverse
problem is the posterior measure, $\postm = \GM{\iparpost}{\C_\text{post}}$ with
$\iparpost$ and $\C_\text{post}$ as given in~\eqref{equ:mean-cov}.
The posterior mean $\iparpost$ is characterized as the minimizer of
\[
\begin{aligned}
& \mathcal{J}(\ipar) :=
  \frac{1}{2} \left\| \mathcal{B}u(\ipar) -\obs  \right\|^2_{\Gamma^{-1}_{\mathrm{noise}}}
  + \frac 12 \left\| \Acal(\ipar - \ipar_0 \right)\|^2_{L^2(\D)},
\end{aligned}
\]
which can also be interpreted as the regularized functional to be
minimized in deterministic inversion.  Next, we specify the action of
the adjoint $\iFF^*$ of the parameter-to-observable map. Given an
observation vector $\dd \in \R^q$ as defined
in~\eqref{equ:param-to-obs}, $\iFF^*\dd$ is computed by solving the
\emph{adjoint equation} (see~\cite{AkcelikBirosDraganescuEtAl05,
  FlathWilcoxAkcelikEtAl11, PetraStadler11}) for the adjoint variable
$p = p(\vec{x}, t)$,
\begin{equation}\label{eq:ad:adj}
  \begin{aligned}
    -p_t - \nabla \cdot (p \vec{v}) - \kappa\Delta p  &= -\obsop^* \dd %
&\quad&\text{ in
    }\D\times (0,T),\\
    p(\cdot, T) &= 0 &&\text{ in } \D,  \\
    (\vec{v}p+\kappa\nabla p)\cdot \vec{n} &=  0 &&\text{ on }
    \partial\D\times (0,T),  
  \end{aligned}
\end{equation}
and setting $\iFF^*\dd = p(\cdot, 0)$.  Note that \eqref{eq:ad:adj} is a
final value problem, which has to be solved backwards in time.

\subsection{Discretization and solvers}
The discretization of the forward and adjoint problems uses
linear triangular/tetrahedral continuous Galerkin finite elements
for two/three space dimensions, and the implicit Euler method in time.  The discrete adjoint
equation is obtained as the adjoint of the discretized forward equation, i.e.,
we follow a discretize-then-optimize approach.  Due to the large
diffusion parameter $\kappa$ used, no stabilization of the
advection-diffusion equation is needed.
A factorization of the matrix in each
implicit Euler time step is computed upfront, such that every time
step (of the forward as well as the adjoint equation) only requires
triangular solves.

The OED optimization problems are solved using \MATLAB's built-in
interior-point solver \verb+fmincon+, to which we supply the OED
objective function evaluation and its derivative, which we compute
using the methods presented above. The optimization algorithm
approximates second derivatives using the
Broyden-Fletcher-Goldfarb-Shanno (BFGS) \cite{NocedalWright06} method.

\section{Numerical results}
\label{sec:numerics}
In this section, we summarize our numerical experiments.  In
Section~\ref{sec:numerics-2d}, we provide a detailed numerical study
of our method and compare different sparsifications for the
two-dimensional model problem. Then, in Section~\ref{sec:numerics-3d},
we use the three-dimensional model problem to test the scalability and
performance of our OED method.

\subsection{The two-dimensional model problem}\label{sec:numerics-2d}
This section contains numerical experiments for the two-dimensional model problem.
We use $\Ns = 122$ candidate locations for placing sensors as shown
in Figure~\ref{fig:ad_dom}(b), and use $1864$ triangles to discretize
the domain. Linear finite elements are used for the parameter $\ipar$
and the state variable $u$ resulting in $\Nm = 1012$ spatial degrees
of freedom (and thus inversion parameters).  The final time is $T=4$
and the time interval is discretized using $\Nt+1$ implicit Euler time
steps; unless otherwise specified, $\Nt = 64$. Observations at the
sensor locations are taken at $\Ntau = 19$ equally spaced points in
the time interval $[1, 4]$.
The prior covariance is specified as described in
Section~\ref{sec:bayes} with $\alpha = 8\times 10^{-3}$ and $\beta =
10^{-2}$; these values result in a Gaussian prior that is not overly
restrictive or smoothing.  Next, we study the low-rank
approximations of the forward operator and the misfit Hessian.

\subsubsection{Low-rank approximation of $\FF$, $\FT$ and $\HM$}
\label{subsec:lowrankF}
We use a randomized SVD to compute a low-rank surrogate for
the parameter-to-observable 
map $\FF:\R^\Nm \to \R^q$ and its prior-preconditioned counterpart $\FT$,
where $q = \Ns\times \Ntau$. This directly leads to a 
low-rank approximation of the prior-preconditioned misfit
Hessian, $\HT(\vec w) = \FT^*\WW\FT$, where the diagonal weight matrix $\WW$ depends
on the design. Since the diagonal entries of $\WW$
may vanish, the numerical rank of $\HT(\vec w)$ is less than or equal
than that of $\FT^*\FT$, which corresponds to assigning unit weights to
all sensors; thus, it suffices to examine the accuracy of the
approximation of $\HT = \FT^*\FT$ only.

First, we study the low-rank approximation of the (prior-preconditioned) 
parameter-to-observable map.  In
Figure~\ref{fig:decomp} (left), we show the normalized singular values
of $\FT$ and of $\F$. The eigenvalues decay rapidly and thus both
mappings can be accurately approximated by low-rank matrices.
Note the faster decay of the singular values of
$\FT$ compared to those of $\F$, which shows that due to the
preconditioning with $\prcov^{1/2}$, $\FT$ can be approximated more
efficiently with a low-rank operator than~$\F$.
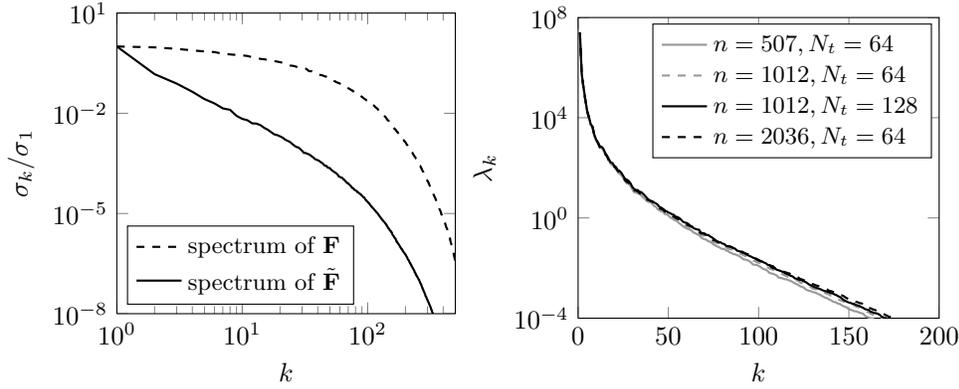
\begin{figure}[ht]\centering
\hspace{-2mm}
\begin{tikzpicture}[]
\begin{loglogaxis}[compat=newest, width=4.5cm, height=4cm, scale only axis,
    xlabel=$k$, ylabel=$\sigma_k / \sigma_1$, xmin=1, xmax=500,
    ymin=1e-8, ymax=1e1,
    legend style={font=\small}, legend pos=south west]
\addplot[color=black,thick,dashed,mark=none,mark size=.25pt]
table[x=n,y=s] {Fdata.txt};
\addlegendentry{spectrum of $\mathbf{F}$}
\addplot[color=black,mark=none,thick,mark size=.25pt]
table[x=n,y=s] {Ftdata.txt};
\addlegendentry{spectrum of $\tilde{\mathbf{F}}$}
\end{loglogaxis}
\end{tikzpicture}
\begin{tikzpicture}[]
\begin{semilogyaxis}[compat=newest, width=4.8cm, height=4cm, scale only axis,
    xlabel=$k$, ylabel=$\lambda_k$, xmin=0, xmax=200,
    ymin=1e-4, ymax=1e8,
    legend style={font=\small,nodes=right}]
\addplot[color=black!40!white,thick,mark=none,mark size=.15pt]
table[x=n,y=dr4] {eigendata.txt};
\addlegendentry{$n = 507, N_t = 64$}
\addplot[color=black!40!white,thick,dashed,mark=none,mark size=.15pt]
table[x=n,y=dr3] {eigendata.txt};
\addlegendentry{$n = 1012, N_t = 64$}
\addplot[color=black,thick,mark=none,mark size=.15pt]
table[x=n,y=dr3c] {eigendata.txt};
\addlegendentry{$n = 1012, N_t = 128$}
\addplot[color=black, thick, dashed, mark=none,mark size=.15pt]
table[x=n,y=dr25] {eigendata.txt};
\addlegendentry{$n = 2036, N_t = 64$}
\end{semilogyaxis}
\end{tikzpicture}\\[-1ex]
\caption{Shown on the left are the normalized singular values of $\FT$ and $\FF$.
The right plot depicts the spectrum of the prior-preconditioned misfit
Hessian with unit sensor weights $\HT=\FT^*\FT$ for
discretizations with different numbers $n$ of spatial and $\Nt$ of temporal unknowns.
To compute these spectra,
we used a rank $r=500$ approximation $\FT_r$ to diminish
the influence of the low-rank approximation for $\FT$.}
\label{fig:decomp}
\end{figure}

Next, we investigate the influence of spatial and temporal
discretization on the prior-preconditioned misfit Hessian; the right
plot in Figure~\ref{fig:decomp} shows the eigenvalues of $\HT$ for
different spatial and temporal discretizations.  Note that the
spectra lie almost on top of each other. This indicates that
the discretizations $\HT$ converge to their
infinite-dimensional counterpart,
$\mathcal{\tilde{H}}_{\text{misfit}}:L^2(\D) \to L^2(\D)$ and that all
discretizations studied in Figure~\ref{fig:decomp} resolve the
dominant physics of the problem.

Finally, we study how the number of candidate sensor
locations affects the spectral decay of $\FT$. Increasing the number
of sensors has the potential to increases the information in the data
and thus the numerical rank of $\FT$.  In
Figure~\ref{fig:svd_sensor_refinement}, we plot the singular values of
$\FT$ for different sensor grids. We observe convergence of the
singular value curves as the sensor grids are refined. This is a
consequence of the correlation of the information that can be gained
from neighboring sensors, which is due to the diffusion term
in \eqref{eq:ad}.

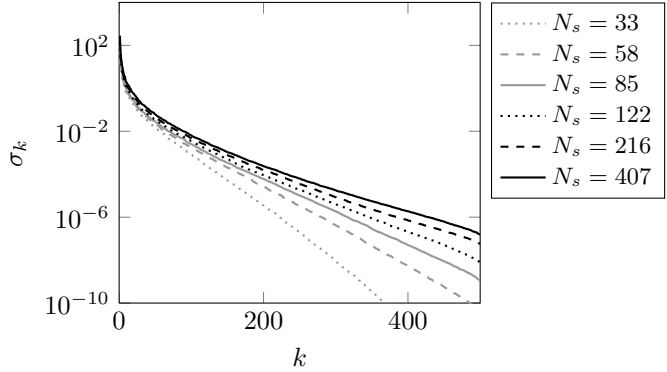
\begin{figure}[ht]\centering
{
\begin{tikzpicture}[]
\begin{semilogyaxis}[compat=newest,width=4.8cm, height=4cm, scale only axis,
    xlabel=$k$, ylabel=$\sigma_k$, xmin=0, xmax=500,
    ymin=1e-10, ymax=1e4,
    legend style={font=\small,nodes=right},legend pos=outer north east]
\addplot[color=black!40!white, thick, dotted]
table[x=idx,y=s] {svds_model_1.txt};
\addlegendentry{$\Ns = 33$}
\addplot[color=black!40!white, thick, dashed]
table[x=idx,y=s] {svds_model_2.txt};
\addlegendentry{$\Ns = 58$}
\addplot[color=black!40!white, thick]
table[x=idx,y=s] {svds_model_3.txt};
\addlegendentry{$\Ns = 85$}
\addplot[color=black!100!white, thick, dotted]
table[x=idx,y=s] {svds_model_4.txt};
\addlegendentry{$\Ns = 122$}
\addplot[color=black!100!white, thick,dashed]
table[x=idx,y=s] {svds_model_6.txt};
\addlegendentry{$\Ns = 216$}
\addplot[color=black!100!white, thick]
table[x=idx,y=s] {svds_model_9.txt};
\addlegendentry{$\Ns = 407$}
\end{semilogyaxis}
\end{tikzpicture}\\[-1ex]
}
\caption{Singular values of $\FT$ for different values of $\Ns$.
The results correspond to a discretization of the problem with
parameter dimension $\Nm = 1012$ and $\Nt = 2^6$ time steps.}
\label{fig:svd_sensor_refinement}
\end{figure}

\subsubsection{A-optimal designs with $\ell_1$-sparsification}\label{sec:sparse-oed-2d}
We now apply our method to compute $\ell_1$-sparsified A-optimal
designs.  The weights $\vec w$ found with $\upgamma = 6\times10^1$ are
shown in Figure~\ref{fig:aopt} (left).  As discussed in
Section~\ref{sec:sparsity}, sensors are then placed at locations
corresponding to non-vanishing weights.  To account for numerical
errors introduced by the interior point method used to
solve~\eqref{equ:oed-generic}, non-vanishing weights are defined as
weights $w_i$ that satisfy ${w_i}/{\sum_j w_j} > 4\times10^{-3}$; the
resulting sensor locations are shown in Figure~\ref{fig:aopt} (right).
Next, we study the behavior of our method for computing
$\ell_1$-sparsified designs.

\def \pos {0.45\columnwidth}
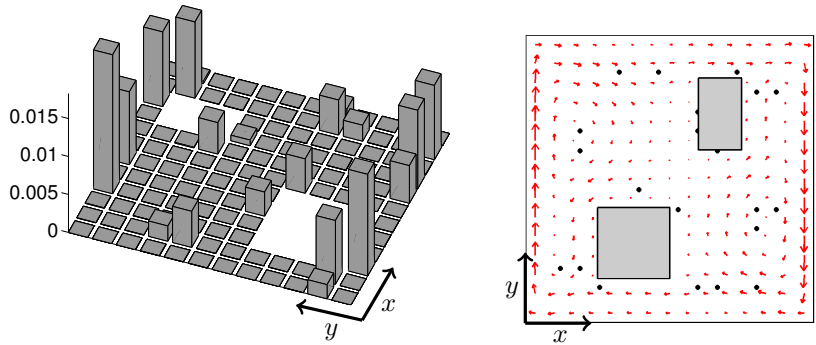
\begin{figure}[ht]\centering
\begin{tikzpicture}
  \node (11) at (0, 0.3) {
  \includegraphics[width=\pos]{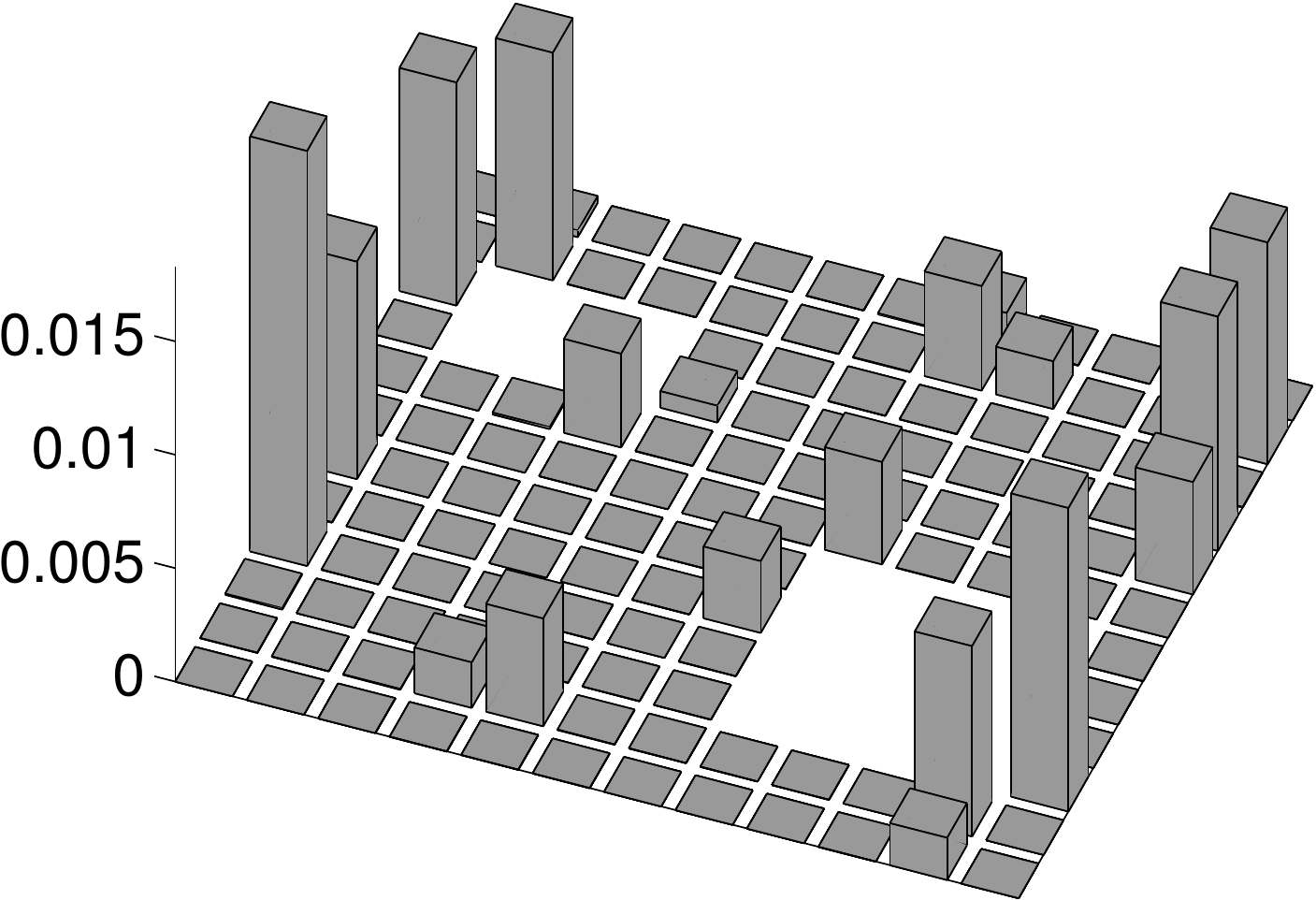}};
  \node (12) at (\pos, 0) {
\includegraphics[width=.3\textwidth]{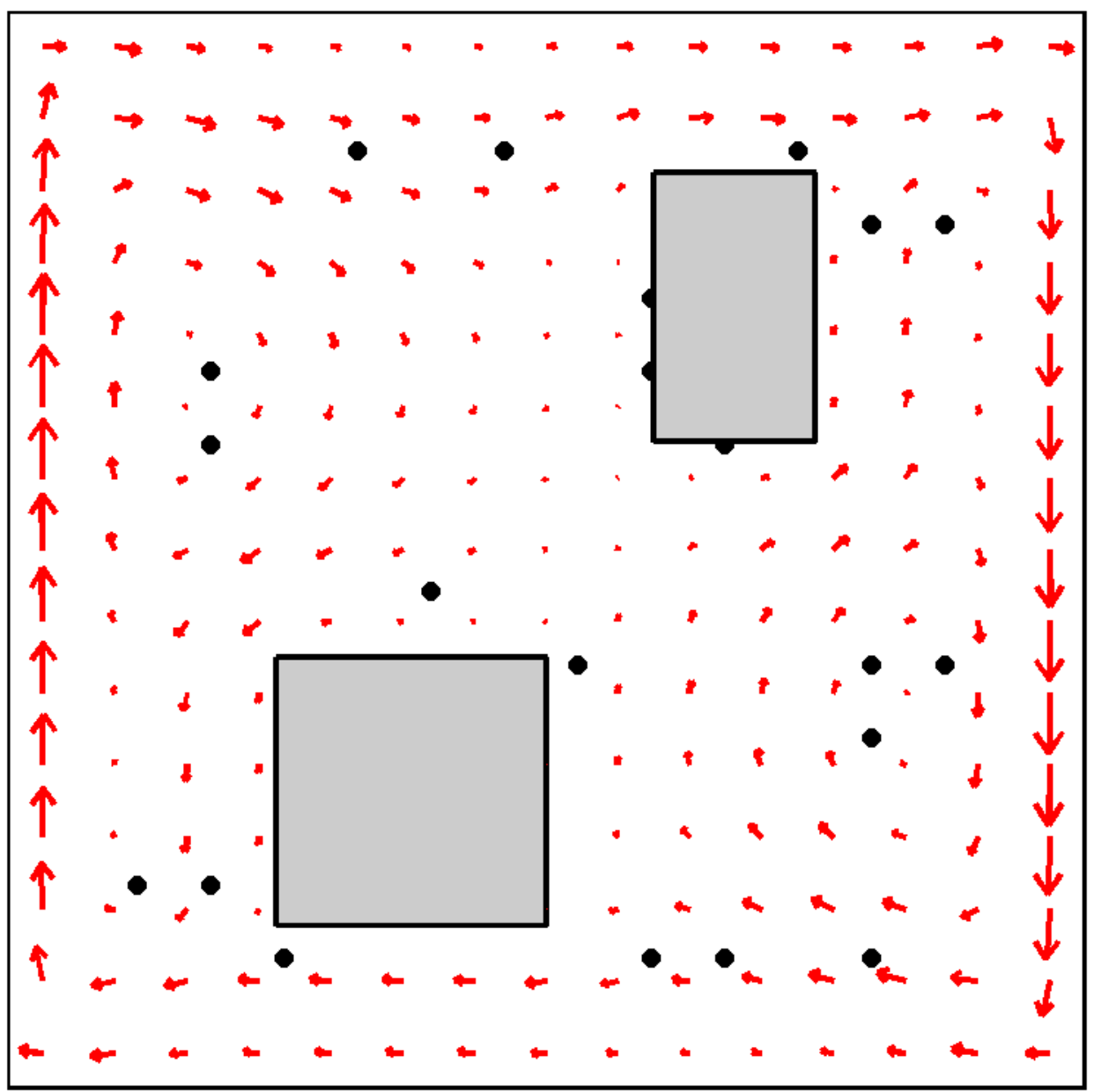}};
 \draw[black,very thick,->] (0.3*\pos,-1.9) -- (0.37*\pos,-1.2);
 \draw[black,very thick,->] (0.3*\pos,-1.9) -- (0.15*\pos,-1.7);
 \draw[black,very thick,->] (0.67*\pos,-1.94) -- (0.82*\pos,-1.94);
 \draw[black,very thick,->] (0.67*\pos,-1.94) -- (0.67*\pos,-1);
 \node at (0.75*\pos,-2.12) {\textcolor{black}{$x$}};
 \node at (0.64*\pos,-1.5) {\textcolor{black}{$y$}};
 \node at (0.36*\pos,-1.7) {\textcolor{black}{$x$}};
 \node at (0.23*\pos,-2.05) {\textcolor{black}{$y$}};
\end{tikzpicture}\\[-1ex]
\caption{Shown on the left are the weights found by solving the
  optimization problem \eqref{equ:oed-generic} with $\ell_1$-sparsification.
  The right image shows the optimal sensor locations (black dots) 
  obtained by placing sensors at locations with non-vanishing 
  weights. Also shown is the advective flow field $\vv$ (red arrows).}
\label{fig:aopt}
\end{figure}

Since we rely on a low-rank surrogate $\FT_r$ for the
prior-preconditioned parameter-to-observable map $\FT$ to avoid
repeated PDE solves in the OED method, we first examine the influence
of this truncation on the optimal design. For that purpose, we fix
$\upgamma=5\times 10^1$ and solve the A-optimal design problem using
approximations $\FT_r$ with $r = 10, 15, 20, 30, 40, 60, 80$.  In
Figure~\ref{fig:rnkstudy}, we plot the optimal objective value
$\Theta(\wopt)$ computed using low-rank surrogates $\FT_r$ with
different $r$.  The convergence of the objective value shows that
very similar results are found for low-rank surrogates $\FT_r$
with $r\ge 40$. This shows that even a significant compression of the
(preconditioned) parameter-to-observable map has little influence on
the optimal design.
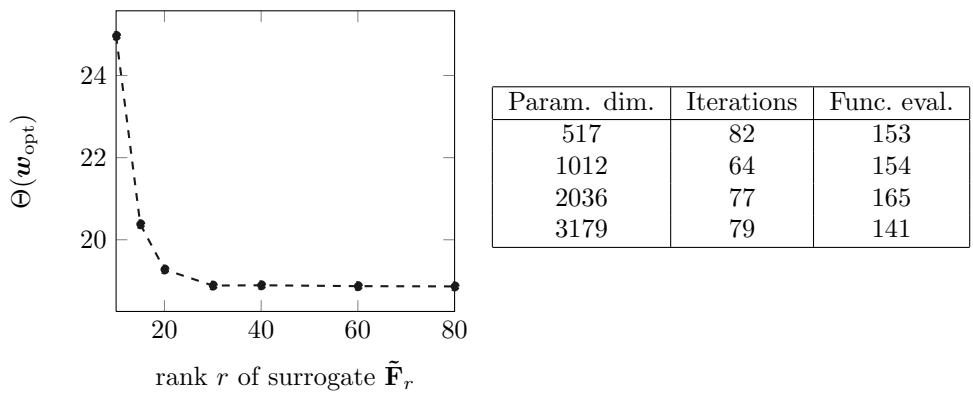
\begin{figure}\centering
\begin{minipage}{.5\textwidth}
  \begin{tikzpicture}
    \begin{axis}[
      width=4.5cm,height=4cm,scale only axis, xmin=10,xmax=80,
      xlabel=rank $r$ of surrogate $\mathbf{\tilde{F}}_r$,
      ylabel=$\Theta(\wopt)$]
    \addplot[color=black!90!white, thick, mark=*,mark size=1.5pt,dashed] coordinates {
   (10.0000,   24.9663)
   (15.0000,   20.3784)
   (20.0000,   19.2772)
   (30.0000,   18.8843)
   (40.0000,   18.8918)
   (60.0000,   18.8703)
   (80.0000,   18.8642)
  (100.0000,   18.8678)};
  \end{axis}
\end{tikzpicture}
\end{minipage}%
\begin{minipage}{.5\textwidth}
  \centering
  \begin{tabular}{|c|c|c|}
    \hline
    Param. dim. & Iterations & Func.~eval. \\
   \hline
     517 & 82 & 153 \\
    1012 & 64 & 154 \\
    2036 & 77 & 165 \\
    3179 & 79 & 141 \\
   \hline
  \end{tabular}
  \vspace{10mm}
\end{minipage}%
\caption{Left: Value of $\Theta(\wopt)$ for an $\ell_1$-sparsified
  A-optimal design obtained with rank-$r$ approximations
  $\FT_r$. Right: Number of interior-point quasi Newton iterations and
  function evaluations for different discretizations (and thus
  different parameter dimensions).}
\label{fig:rnkstudy}
\end{figure}

Next, we study the effect of increasing the parameter dimension
(through mesh refinements) on the number of interior-point iterations
required to solve the OED optimization problem \eqref{equ:oed-est}. As
observed in Section~\ref{subsec:lowrankF}, increasing the parameter
dimension does not increase the numerical rank of $\FT$ once the mesh
is sufficiently fine. This, and the use of a Newton-type optimization
method results in a nearly constant number of interior-point
iterations required to solve \eqref{equ:oed-est}, as shown in
the table in Figure~\ref{fig:rnkstudy}.

Finally, we study how the performance of the numerical optimization encountered in our OED method
is affected by increasing the number $\Ns$ of candidate
sensor locations; note that $\Ns$ is also the dimension of the weight
vector $\vec w$, the unknown in the OED optimization problem
\eqref{equ:oed-est}. The results shown in
Table~\ref{tbl:ip_sensor_refinement} show that the number of
interior-point quasi-Newton iterations and the number of evaluations
of $\Theta(\cdot)$ do not increase significantly as the number of
sensor candidate locations increases. We again attribute this to the
use of a Newton-type method, since Newton's method satisfies---under
certain assumptions---a mesh independence property \cite{Deuflhard04}.

\begin{table}[t!]\centering
\caption{Number of interior-point BFGS iterations and evaluations of $\Theta(\cdot)$
  required for the solution of \eqref{equ:oed-est} with $\ell_1$-sparsification,
  for different numbers $\Ns$ of sensor candidate locations. The
  optimization iteration was terminated when the norm of the
  gradient was decreased by a factor of $10^{4}$.}
\label{tbl:ip_sensor_refinement}
\begin{tabular}{|l|cccccccc|}
\hline
$\Ns$      &  $33$ &  $58$  &  $85$ &  $122$ & $172$ &  $216$ &  $264$ &  $340$\\
\hline
Iterations &  $58$  &  $68$ &   $68$ &   $81$ &   $58$ &   $68$ &   $75$  &  $72$ \\
\hline
Func.~eval. &  $149$ &  $165$ &  $141$ &  $158$ &  $136$ &  $137$ &  $122$ &  $160$ \\
\hline
\end{tabular}
\end{table}

Regarding the complexity of our method with respect to the number $\Ns$
of candidate sensor locations, we draw the following conclusions from
Figure~\ref{fig:svd_sensor_refinement} and
Table~\ref{tbl:ip_sensor_refinement}:
First, the number of PDE solves required to compute the low-rank 
surrogate $\FT_r$ is bounded with respect to $\Ns$, 
and second, the number of optimization iterations to solve
\eqref{equ:oed-est} is insensitive to the number of candidate sensor
locations, and thus the dimension of the weight vector $\vec w$.

\subsubsection{A-optimal designs with $\Phi_\eps$-sparsification}
Next, we present $\Phi_\eps$-sparsified designs obtained by a
continuation procedure with respect to $\eps$.
Optimal designs for $\upgamma = 0.05$ and $\upgamma=0.06$ are shown on the
left in Figure~\ref{fig:aopt-pe}.  Note that the optimal sensor
locations obtained with $\upgamma=0.06$ is not a subset of the
locations obtained with $\upgamma=0.05$, i.e., the designs are not
nested.

For the continuation procedure, we use the values $\eps_i = (2/3)^i$,
$i = 1, \ldots, 10$.
To elucidate the convergence of the weights to the desired binary 
structure, the weights corresponding to several values of $\eps_i$ are
shown on the right in Figure~\ref{fig:aopt-pe}. Note that the behavior
of the weights is not monotone as $\eps$ is decreased. A possible
explanation for this behavior is that, as $\eps$ decreases, weights at
neighboring locations can get merged into a single weight at a new
location.

\begin{figure}[ht]\centering
\begin{minipage}[t]{.4\columnwidth}
\vspace{0pt}
\includegraphics[width=.85\textwidth]{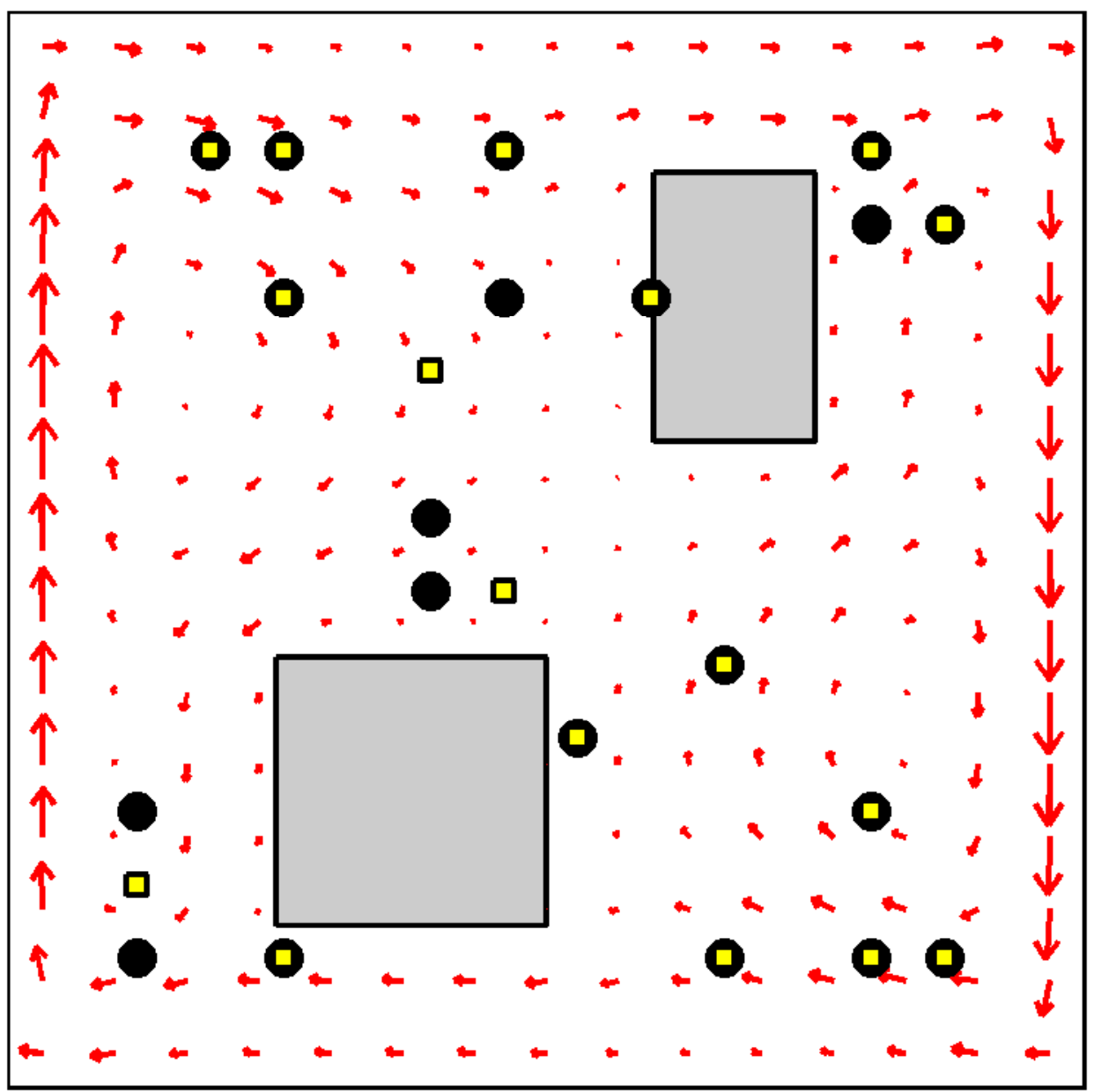}
\end{minipage}
\hspace{-2ex}
\begin{minipage}[t]{.5\columnwidth}
\vspace{0pt}
\begin{tikzpicture}[]
\begin{axis}[compat=1.4, width=4.7cm, height=4.2cm, scale only axis,
    xlabel=$j$, ylabel=$w_{s(j)}$, xmin=1, xmax=122,domain=1:122,
    ymin=0, ymax=1.1,
    legend style={font=\small,nodes=right}, legend pos=south east]
\addplot [color=red, mark=*, mark size=1pt,only marks]
table[x=idx,y=w] {contweights1.txt};
\addlegendentry{$\eps_1$}
\addplot [color=green, mark=*, mark size=1pt,only marks]
table[x=idx,y=w] {contweights2.txt};
\addlegendentry{$\eps_2$}
\addplot [color=blue, mark=*, mark size=1pt,only marks]
table[x=idx,y=w] {contweights3.txt};
\addlegendentry{$\eps_4$}
\addplot [color=black, mark=*, mark size=1pt,only marks]
table[x=idx,y=w] {contweights4.txt};
\addlegendentry{$\eps_{10}$}
\addplot[mark size=0.2,mark=none,dashed]{1};
\draw[dashed,color=black] (41,0)--(41 ,4.3 cm);
\end{axis}
\end{tikzpicture}\hspace{2ex}
\end{minipage}
\vspace{-2ex}
\caption{Left: Sensor locations using $\Phi_\eps$-sparsified A-optimal
  design for $\upgamma = 0.05$ (black dots) and for $\upgamma =
  0.06$ (yellow squares). Right: Convergence of the weights to a 0--1
  structure as $\eps\to 0$.  The weights are ordered based on the
  values obtained with $\eps_1$, and shown is the evolution of the
  weights as $\eps$ decreases.  Here, $s(j)$ denotes the descending
  ordering obtained with $\eps_1$.}
\label{fig:aopt-pe}
\end{figure}
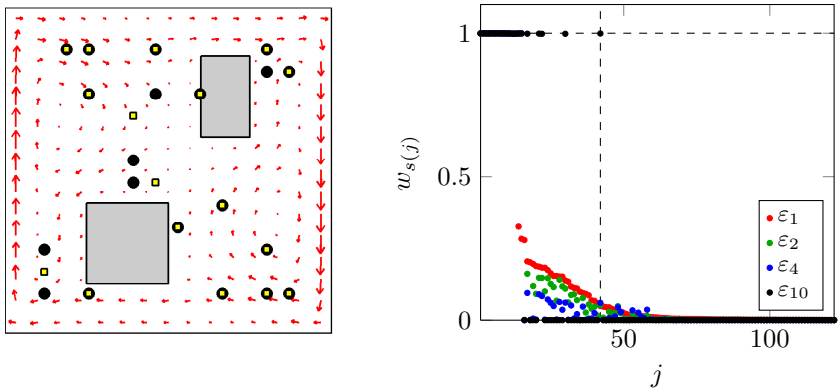

To illustrate the decreased variance for A-optimal designs,
in Figure~\ref{fig:variance} we compare the pointwise posterior
standard deviation obtained with
$\upgamma=0.05$
with a uniform and two randomly generated
designs employing the same number of sensors. This comparison shows that the A-optimal
design provides a notable reduction in the standard deviation (and
thus the variance) fields.
Next, we compare the performance of sensor placements
obtained with $\ell_1$- and $\Phi_\eps$-sparsification.
\begin{figure}[ht]\centering
\includegraphics[width=.242\textwidth]{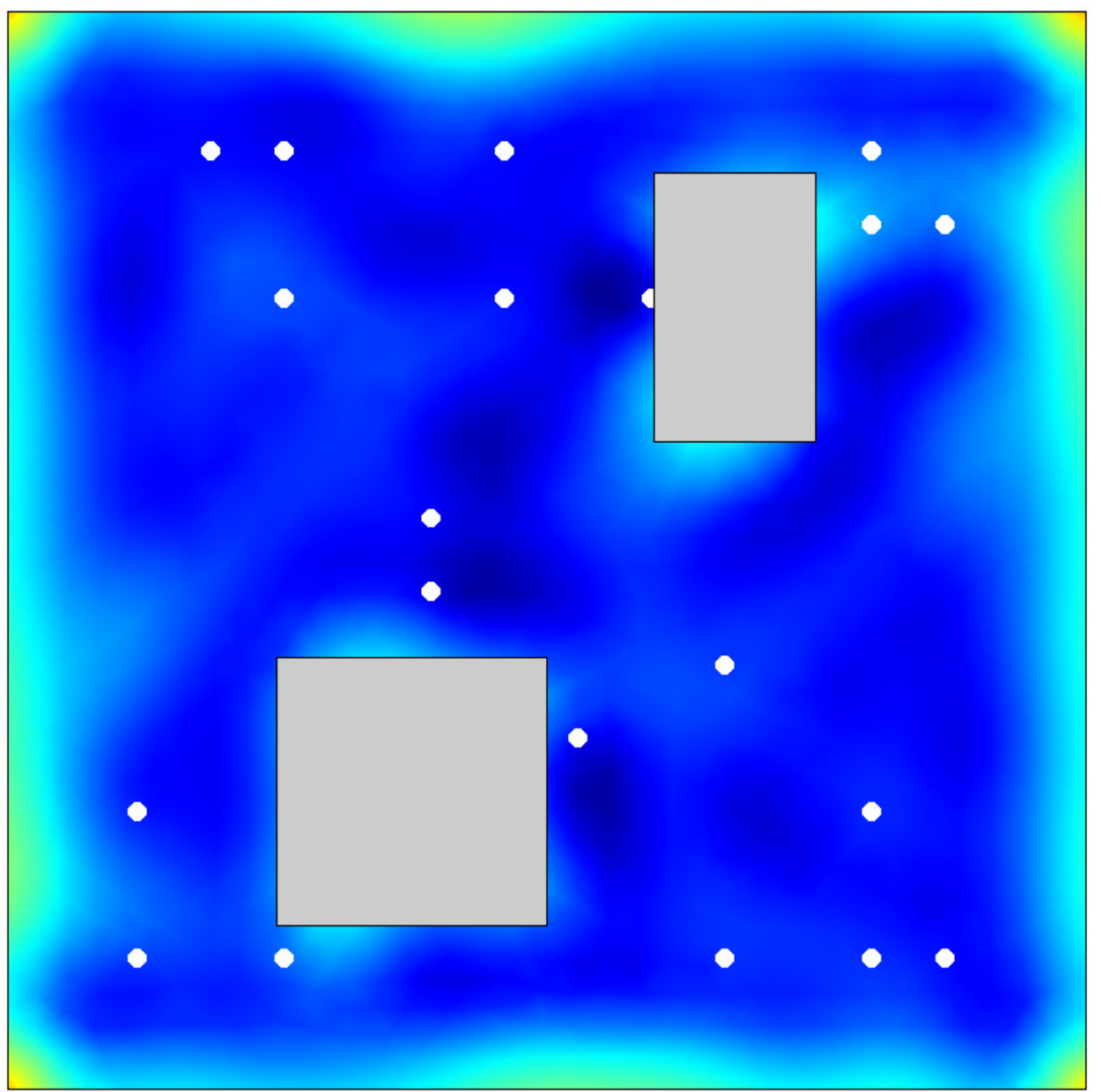} %
\hfill
\includegraphics[width=.242\textwidth]{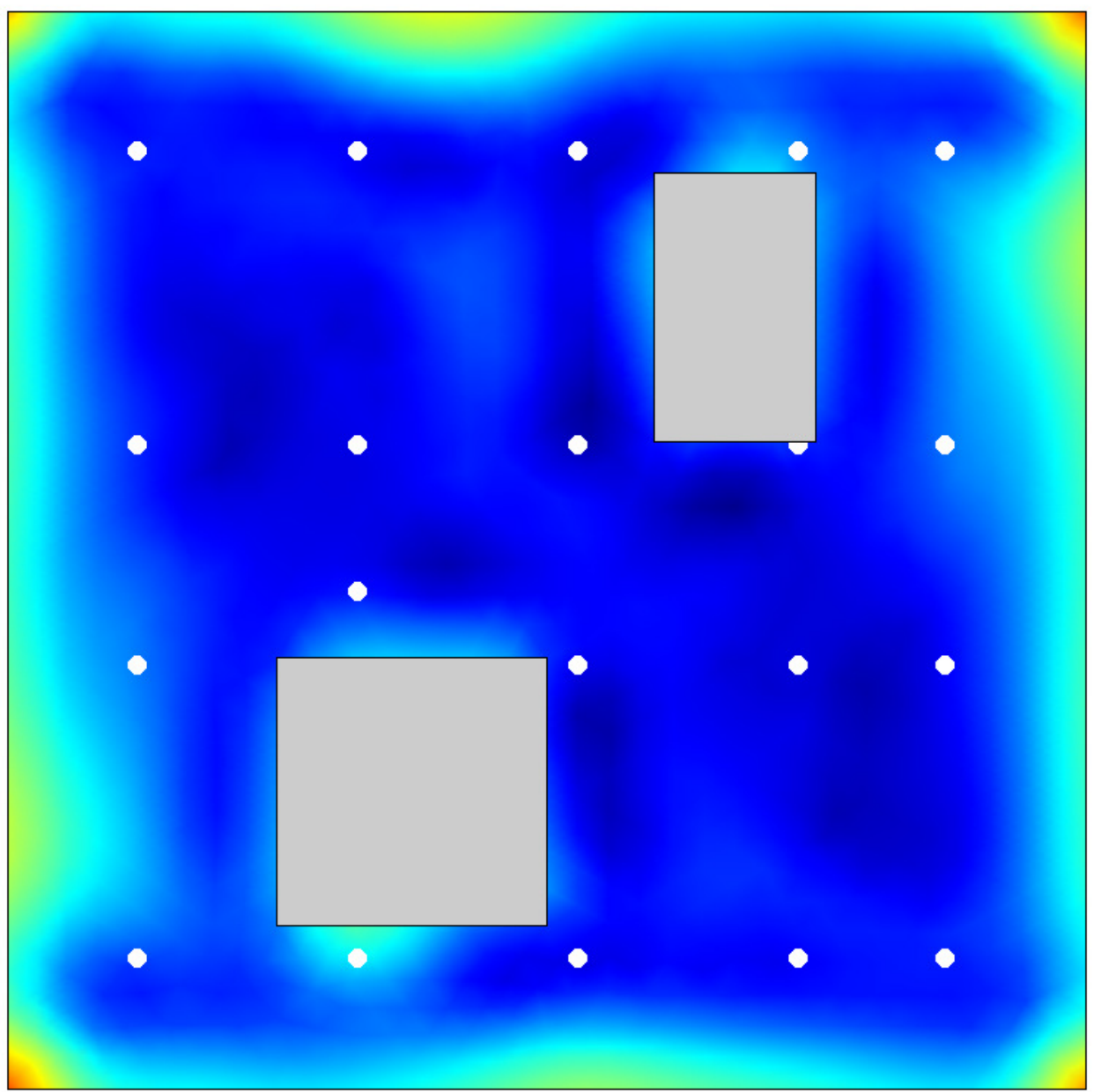} %
\hfill
\includegraphics[width=.242\textwidth]{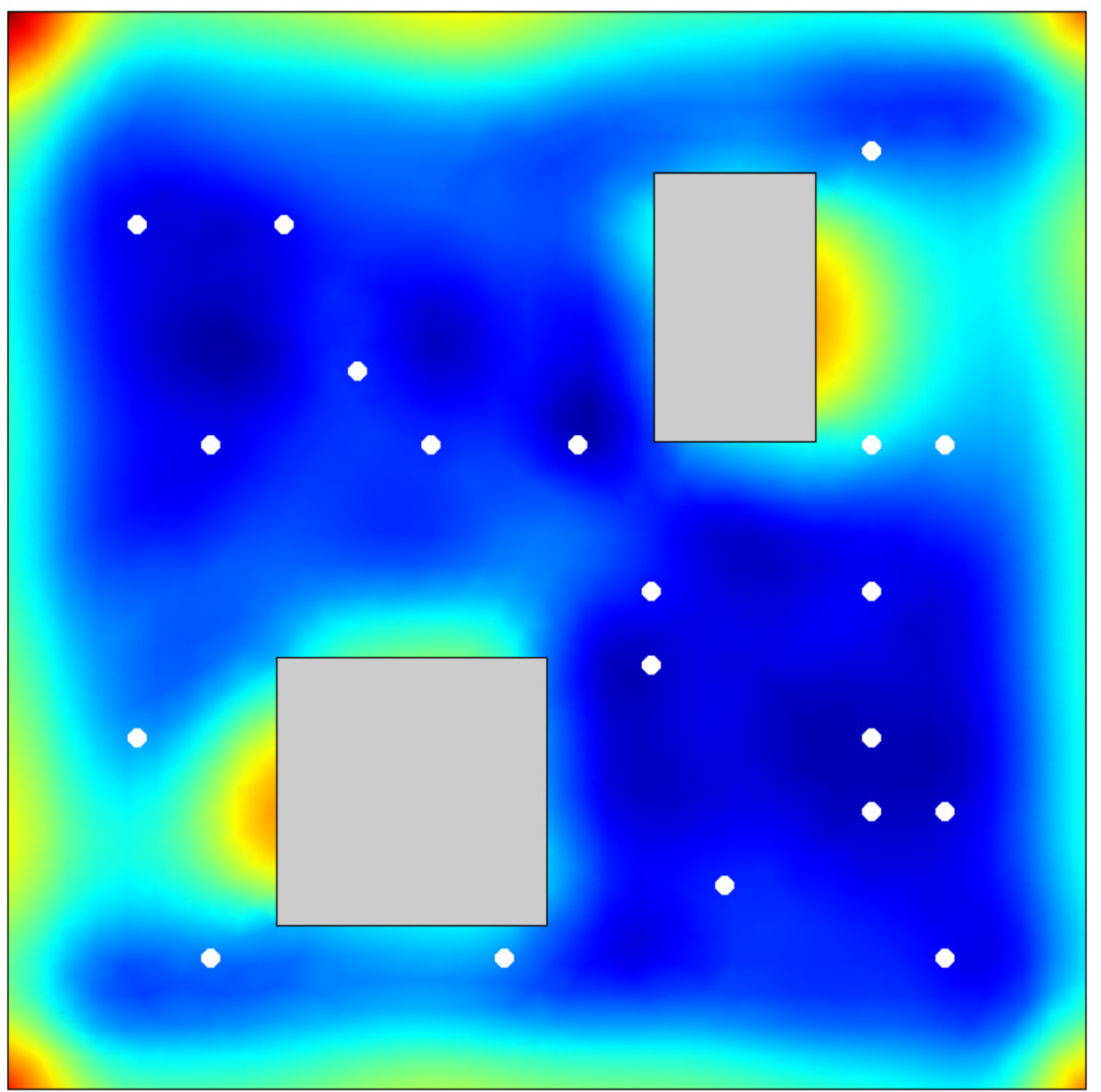} %
\hfill
\includegraphics[width=.242\textwidth]{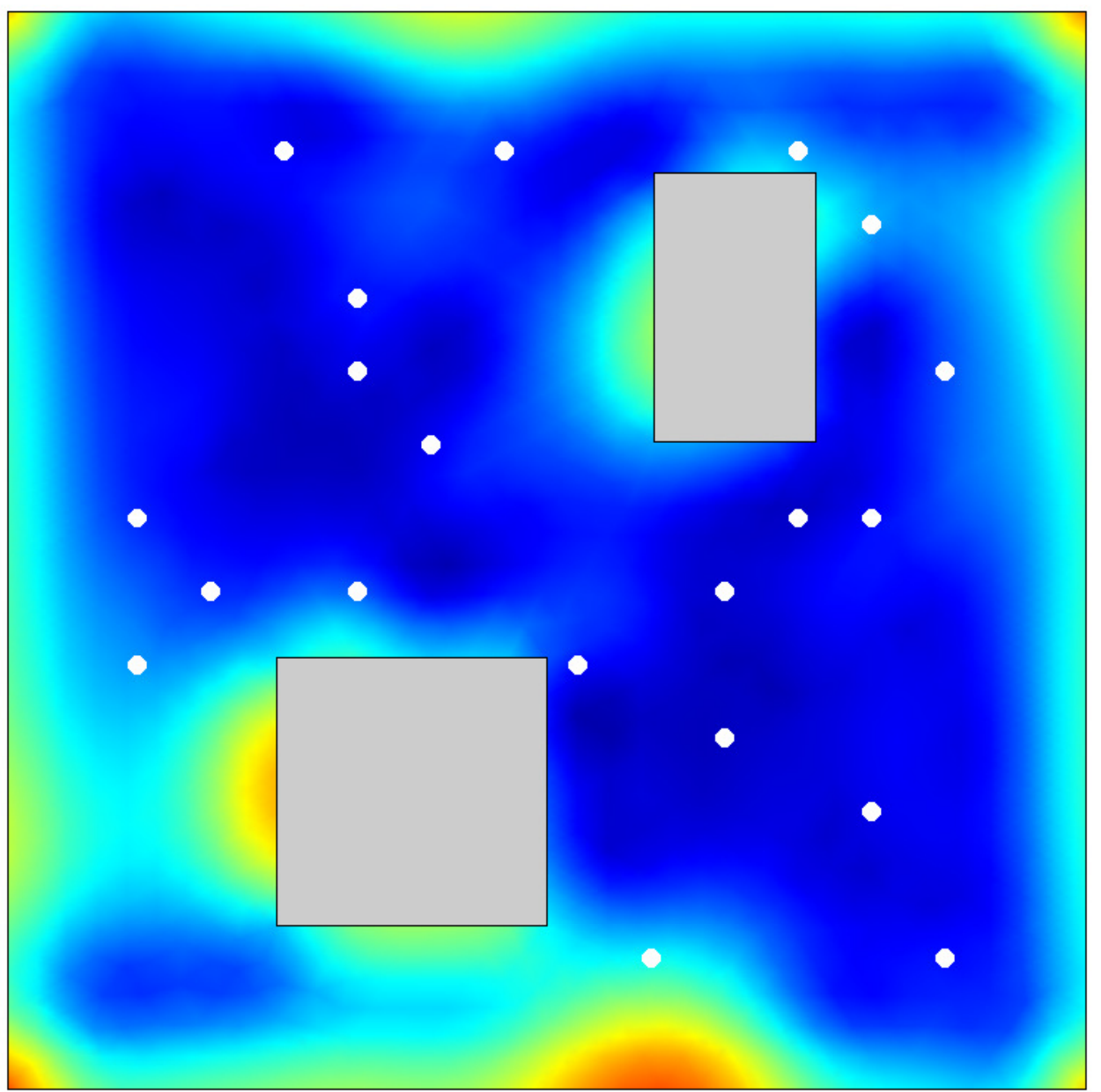}\\[-.5ex]
\hspace{1.2cm} (a) \hfill (b) \hfill (c) \hfill (d) \hspace{1.2cm} \phantom{.} \\[-1ex]
\caption{Shown is the pointwise standard deviation of the posterior
  for a $\Phi_\eps$-sparsified A-optimal design with 20 sensors (a),
  for a manually chosen uniform design (b), and two randomly generated
  designs with the same number of sensors (c), (d).  The white
  dots indicate the sensor locations. Red and blue correspond to regions with large
  and small standard deviation, respectively.
  We find that compared to the optimal design, the uniform design and the two random designs result in
  $7\%$, $36\%$ and $26\%$ increase in the average variance,
  respectively.}
\label{fig:variance}
\end{figure}

\subsubsection{Comparing $\ell_1$- and $\Phi_\eps$-sparsifications}
We observe that $\ell_1$- and $\Phi_\eps$-sparsifications lead to
different designs for the same number of sensors. Thus, a naturally
arising question is (1) \emph{which sparsification results in the
  better placement of sensors?}  Moreover, (2) how much do
\emph{optimal} designs improve over \emph{randomly} chosen designs?
To answer these questions,
we compute optimal designs based on different sparsification
strategies and compare with randomly generated designs. 
We use trace estimators with $100$ Gaussian random vectors chosen
differently in each OED problem, but report the exact value of
$\trace(\matrix{\Gamma}_\text{post}(\vec w)) = \trace(\H(\vec
w)^{-1})$.

Our results are summarized in Figure~\ref{fig:super-plot}.
We compute $\ell_1$-sparsified designs for various values of
$\upgamma$ %
and report the value of $\trace(\H(\vec w)^{-1})$.  Similarly, we
report $\trace(\H(\vec w)^{-1})$ for 0--1 optimal designs computed via
$\Phi_\eps$-sparsification.  Additionally, we compute $\trace(\H(\vec
w)^{-1})$ for a collection of random sensor configurations.  We find
that $\Phi_\eps$-sparsified designs consistently outperform
$\ell_1$-sparsified designs and that both improve
significantly over randomly chosen designs.
Another observation from
Figure~\ref{fig:super-plot} is the diminishing returns as the number
of sensors is increased; using more than $20$ sensors only results in
negligible decrease of the OED objective function value.

\subsubsection{Influence of trace estimation}
To assess the accuracy of randomized trace estimation for a typical
posterior covariance matrix, we compare the exact trace with results
from randomized trace estimation with different numbers of random
vectors. We find that trace estimators based on 1, 5, 10, 20, 100
vectors estimate the exact trace with an average error of about 15\%,
7\%, 5\%, 2\% and 1.5\%, respectively.  This is consistent with the
experiments reported in~\cite{AvronToledo11} showing that
randomized trace estimation is reasonably accurate
with a small number of random vectors, but that many random
vectors may be needed to obtain very accurate approximations.

A more relevant question in the context of OED is how the use of trace
estimators to compute optimal sensor locations influences the
resulting designs. To study this issue, we compute optimal designs
with fixed $\upgamma$, but different trace estimators. The results are
shown in Figure~\ref{fig:trace_sensitivity}. Note that trace
estimation in the OED objective function does have an influence on the
designs. More accurate trace estimation decreases the variation in both
the number of sensors of the designs and the value of the trace of the posterior
covariance.  However, as can be seen from Figure~\ref{fig:super-plot}
optimal designs computed using trace estimation (based on 100 random
vectors) consistently improve over random designs with
respect to reducing the \emph{exact} trace of the posterior covariance.  We
thus conclude that the use of trace estimation does not have a
significant impact on the quality of the optimal designs.

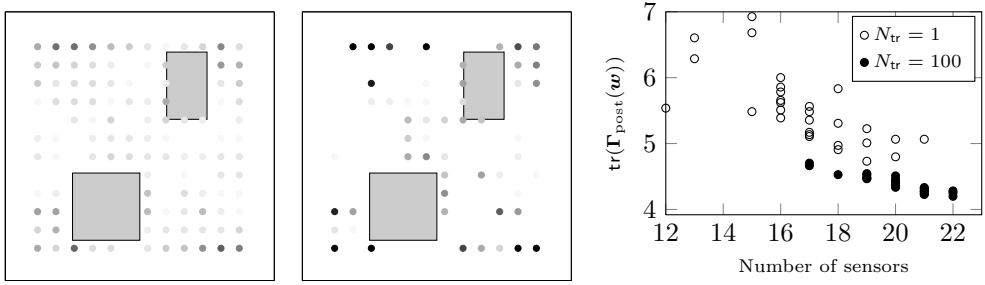
\begin{figure}
\includegraphics[width=.28\textwidth]{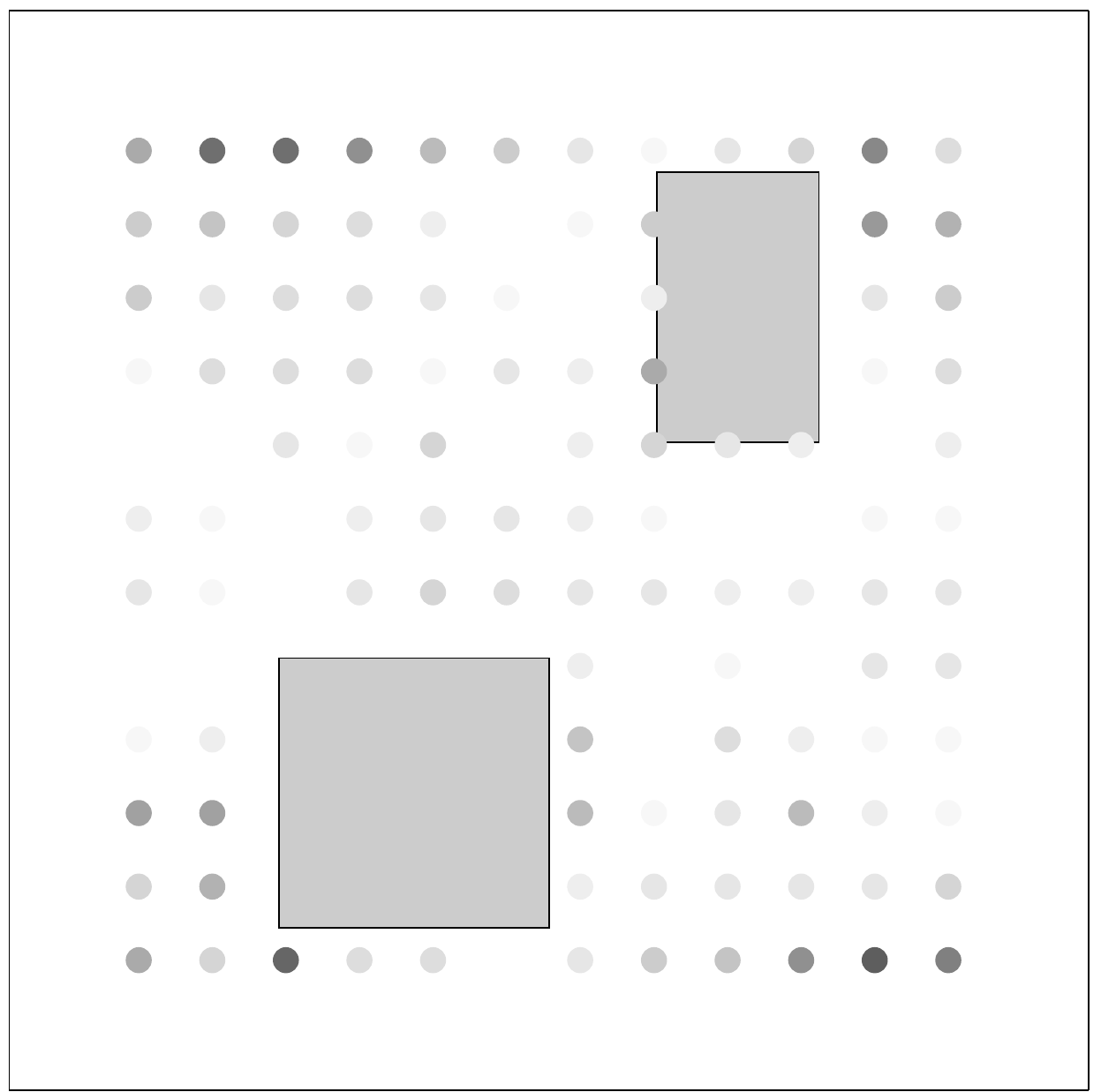}
\hfill
\includegraphics[width=.28\textwidth]{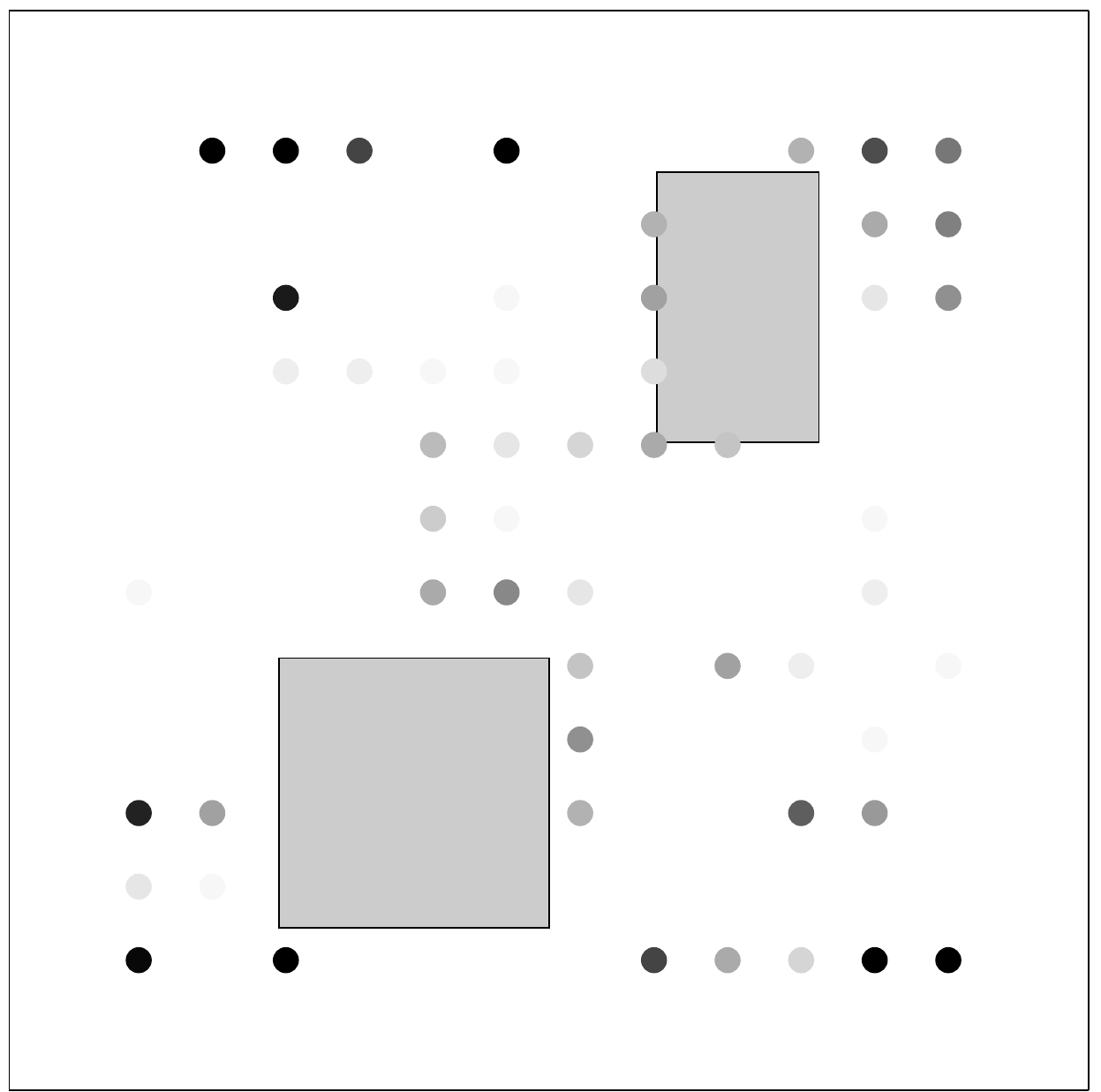}
\hfill
\begin{tikzpicture}[]
\begin{axis}[compat=newest, width=4.2cm, height=2.7cm, scale only axis, xtick={12,14,16,18,20,22}, 
    xlabel={\scriptsize Number of sensors}, ylabel={\scriptsize $\trace(\postcov(\vec w))$},
    xmin=12, xmax=23, ymax=7,
    legend style={font=\scriptsize,nodes=right}]
\addplot [color=black, mark=o, mark size=1.5pt, only marks]
table[x=nnz,y=tr] {freq_data_tr1.txt};
\addlegendentry{$\Ntr = 1$}
\addplot [color=black, mark=*, mark size=1.5pt, only marks]
table[x=nnz,y=tr] {freq_data_tr100.txt};
\addlegendentry{$\Ntr = 100$}
\end{axis}
\end{tikzpicture}
\caption{Sensitivity of $\Phi_\eps$-sparsified designs with respect to
  the trace estimator with fixed $\upgamma=0.05$. The left and the
  middle image visualize the frequency of a candidate location being
  part of the optimal design with a trace estimator based on a single
  random vector (left) and 100 random vectors (middle). The results
  are based on 30 different realizations of the respective trace
  estimators. The darker the candidate location, the more often the
  corresponding sensor was part of the optimal design. Note the
  decreased variation in the designs in the middle plot, which is due
  to the increased accuracy of the trace estimator. The right plot
  depicts the number of sensors and the value of the trace of the
  posterior covariance for 30 different trace estimators with a single
  random vector (empty dots) and with 100 random vectors (filled
  dots). Note that the number of sensors varies for the different
  realizations and that the more accurate trace estimator results in less
  scattering with respect to the number of sensors and the value of
  the trace. \label{fig:trace_sensitivity}}
\end{figure}

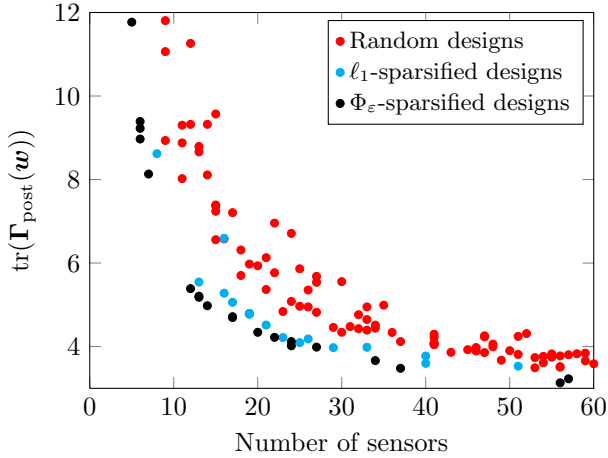
\begin{figure}[ht]\centering
\hspace{-1.1cm}
\begin{tikzpicture}[]
\begin{axis}[compat=newest, width=6.7cm, height=5cm, scale only axis,
    xlabel=Number of sensors, ylabel=tr($\mathbf{\Gamma}_\text{post}({\mbox{\boldmath$\displaystyle{w}$}})$), xmin=0, xmax=60,
    ymin=3, ymax=12,
    legend style={font=\small,nodes=right}]
\addplot [color=red, mark=*, mark size=1.5pt, only marks]
table[x=nnz,y=trace] {randpts.txt};
\addlegendentry{Random designs}
\addplot [color=cyan, mark=*, mark size=1.5pt, only marks]
table[x=nnz,y=trace] {l1pts.txt};
\addlegendentry{$\ell_1$-sparsified designs}
\addplot [color=black, mark=*, mark size=1.5pt, only marks]
table[x=nnz,y=trace] {binpts.txt};
\addlegendentry{$\Phi_\eps$-sparsified designs}
\end{axis}
\end{tikzpicture}
\caption{
  Comparison of performance for different designs.  Shown is
  the value of $\trace(\H^{-1})$ versus the number of sensors for
  random designs (red dots), $\ell_1$-sparsified designs (cyan
  dots) and $\Phi_\eps$-sparsified designs (black
  dots). 
}
\label{fig:super-plot}
\end{figure}

\subsection{The three-dimensional model problem}
\label{sec:numerics-3d}
The main target of this three dimensional model problem is to
study the applicability of our OED method
to large-scale inverse problems.  The
computational domain used is depicted in Figure~\ref{fig:ad_dom}(c), where
$101$ candidate sensor locations are shown as black dots. Note that we
allow sensors on the ground and on the sides of the buildings.
Observations are collected at $\Ntau = 12$ equally spaced points in
the time interval $[1, T]$, where $T=4$.  The parameter and the
state variable $u$ in the advection-diffusion equation are discretized using
linear finite elements on a tetrahedral mesh with $10{,}652$ spatial
degrees of freedom, and $64$ implicit Euler time steps are used for the time
integration.  Thus, the dimension of the inversion parameter %
is $\Nm = 10{,}652$.

\begin{figure}[ht]\centering
\includegraphics[width=.4\textwidth]{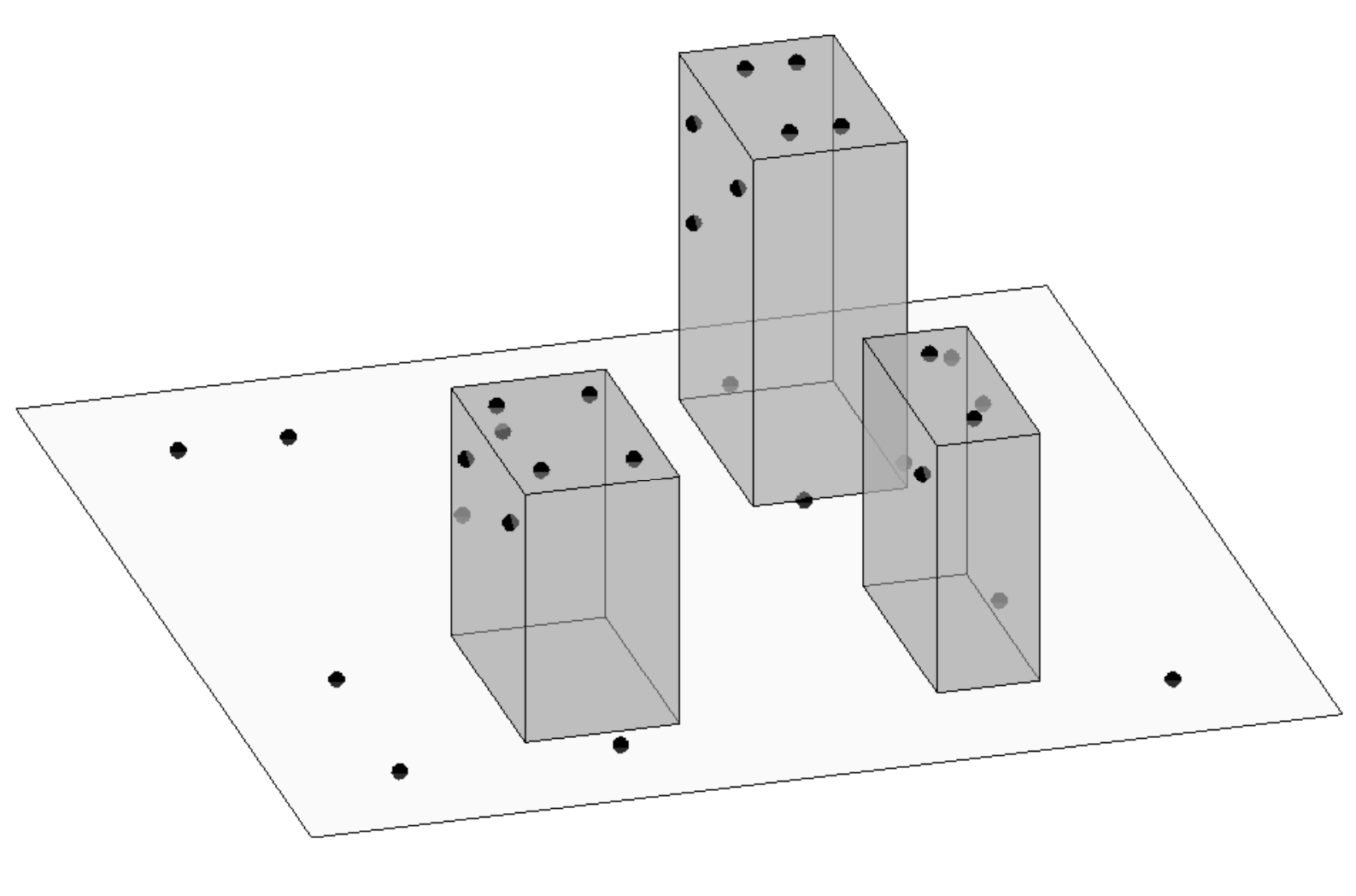}%
\includegraphics[width=.4\textwidth]{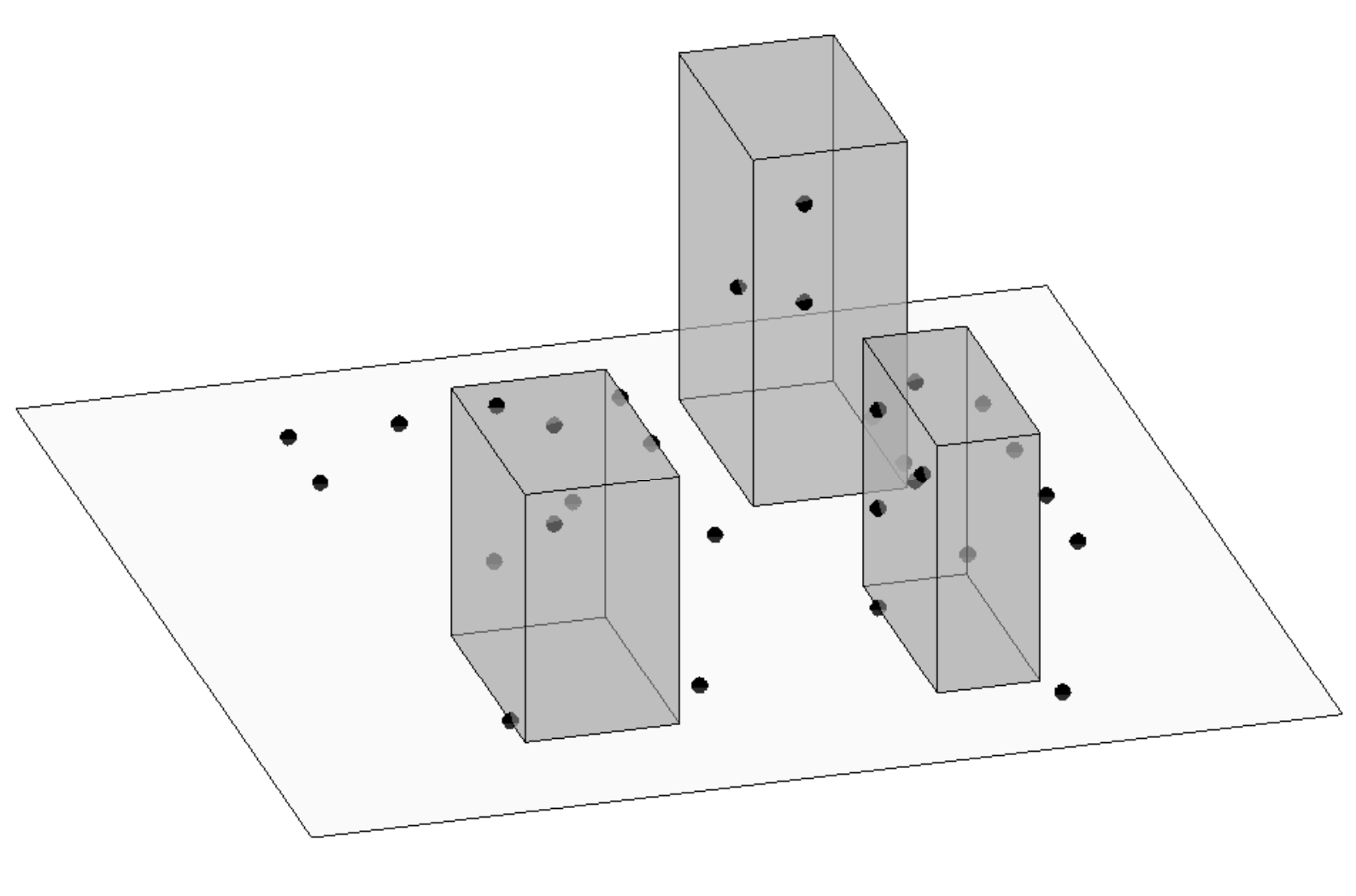}\\
\includegraphics[width=.4\textwidth]{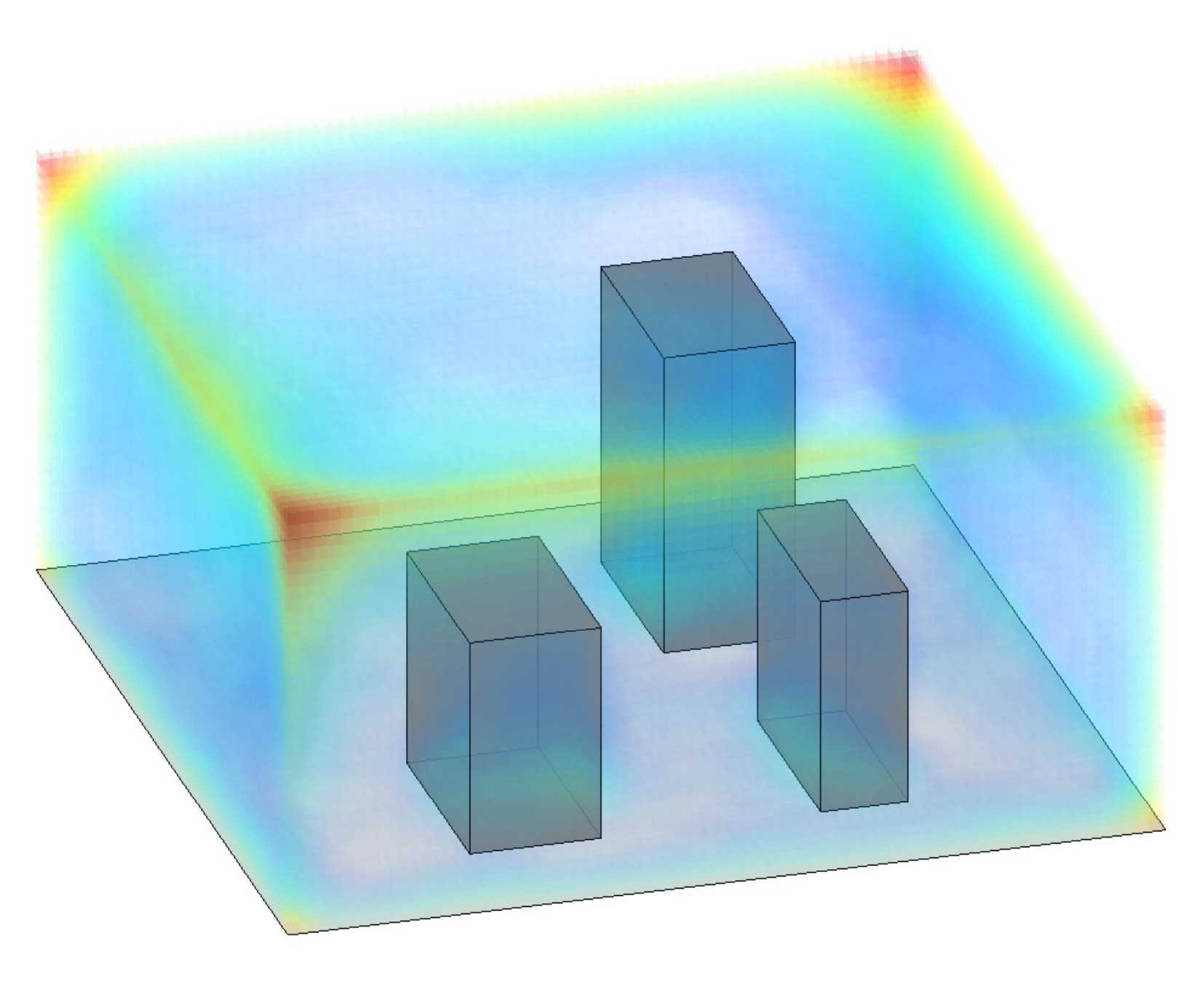}
\includegraphics[width=.4\textwidth]{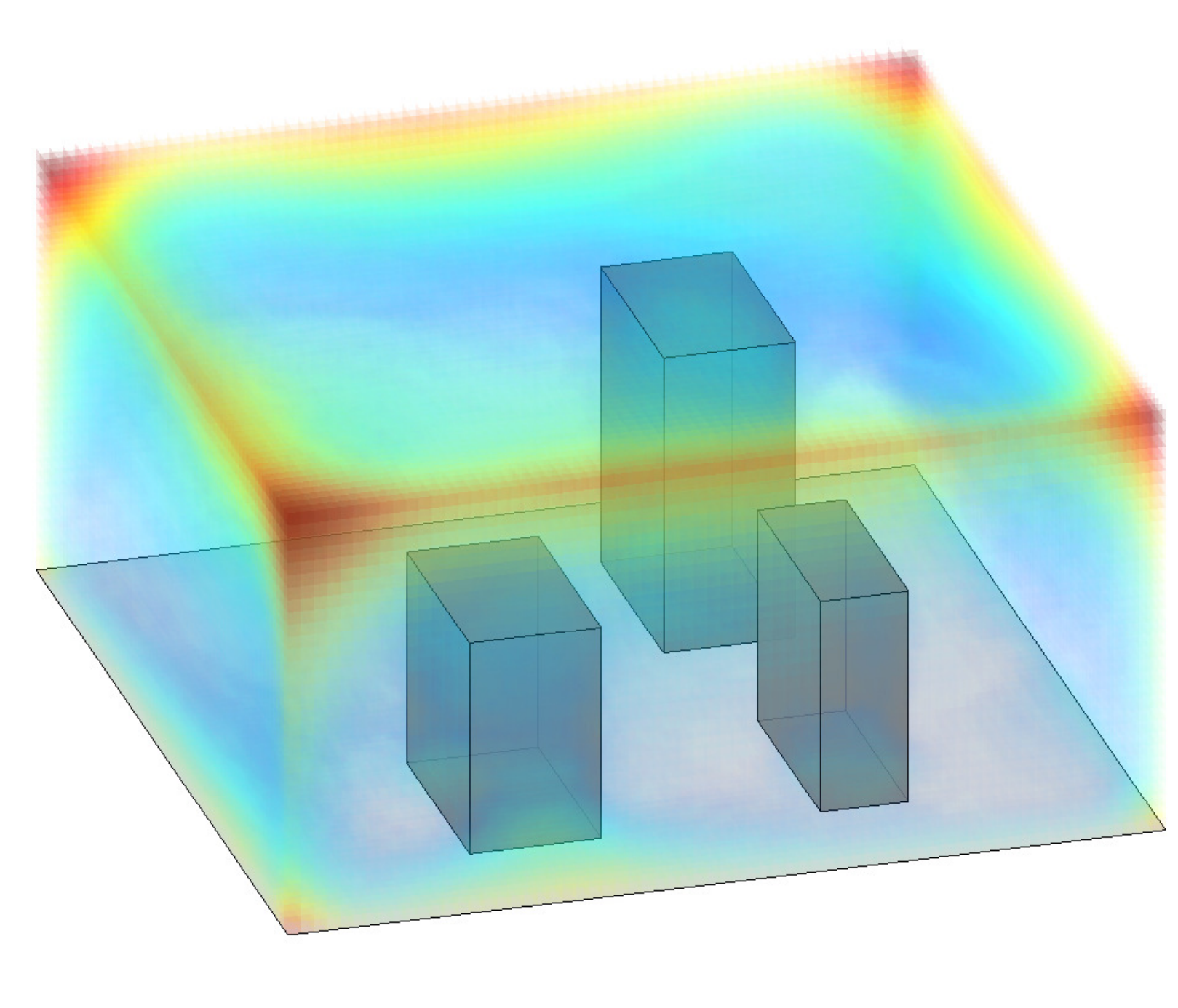}
\vspace{-2ex}
\caption{Top row: Sensor locations (black dots) obtained with A-optimal
  design with $\Phi_\eps$-sparsification (left) and with random design
  (right). Bottom row: Standard deviation fields of the posterior for
  an A-optimal design with $\Phi_\eps$-sparsification (left) and random
  design (right). Here, red indicates regions with large standard
  deviation, and blue corresponds to regions with small standard
  deviation.}
\label{fig:aopt3d}
\end{figure}

In our computations, we found a rank-$100$ SVD surrogate 
for the prior-preconditioned parameter-to-observable map adequate since
the $100$th singular value of $\FT$ is about four
orders of magnitude smaller than the largest singular value.
This translates to eight orders of magnitude difference between the
first and the $100$th eigenvalue of $\HT$.
Due to the large parameter space dimension we do not
compute $\M^{-1/2}$ explicitly, but use the algorithm for the iterative application of
$\M^{-1/2}$ to vectors from \cite{ChenAnitescuSaad11}; see also Section~\ref{sec:aopt-alg}.
We observe that the algorithm is well converged after $10$
applications of $\Mt$---the relative difference between the iterative
application of $\Mt^{-1/2}$ based on $10$ applications of $\Mt$
compared with that based on $500$ applications of $\Mt$ is
approximately $8 \times 10^{-6}$. For the solution of the OED
optimization problem with $\upgamma=1.2\times 10^{-1}$ and
$\Phi_\eps$-sparsification, the continuation procedure presented in
Section~\ref{sec:sparse-oed-2d} is used. To compute the initial
$\ell_1$-sparsified weight vector $\vec{w}^0$, 18 interior-point
iterations were necessary.
This initialization was followed by 10 continuation steps using
the penalties $\Phi_{\eps_i}$, with $\eps_i = (2/3)^i$; 
each step required the solution of an auxiliary OED
problem, which amounted to a total of $554$ interior-point quasi-Newton
iterations to arrive at a 0--1 design vector. The iteration for each
of these auxiliary OED problems was terminated when the relative residual
dropped below $10^{-4}$ or a maximum of 150 iterations was reached.
The number of quasi-Newton iterations in the continuation procedure can
likely be decreased by incorporating the continuation procedure into
the interior-point algorithm rather than solving an independent
interior point problem for each $\eps_i$.

The resulting optimal design is shown in top left image of
Figure~\ref{fig:aopt3d}; note that several sensors are placed on the top
of buildings; these sensors help reduce the variance in the upper
parts of the domain.  To illustrate the effectiveness of the optimal
design to reduce the variance, in Figure~\ref{fig:aopt3d} we show
a volume rendering of the pointwise posterior standard deviation
obtained with the optimal design, and compare it with a rendering obtained
with a randomly generated design that has the same number of
sensors.

\section{Concluding remarks}\label{sec:conc}

We have developed a structure-exploiting efficient numerical method
for computing A-optimal experimental designs in infinite-dimensional
Bayesian linear inverse problems governed by PDEs. Numerical
experiments for the inversion of the initial condition in an
advection-diffusion equation indicate that an optimal design can be
computed at a cost, measured in forward PDE solves, that is
independent of the parameter and candidate sensor dimensions.
Moreover, we find that experimental designs obtained with regularized
$\ell_0$-sparsification are consistently superior to designs obtained
with $\ell_1$-sparsification, and significantly improve over uniform
and random designs.

One limitation of the OED method we have presented is that it relies
on linearity of the parameter-to-observable map and assumes Gaussian
prior and noise distributions. However, the methods presented here are
also applicable in situations where a nonlinear
parameter-to-observable map can be well approximated by a
linearization over the set of parameters that have significant
posterior probability. Moreover, using Gaussian noise and prior
distributions is common, particularly in infinite-dimensional
inference problems.  The computational efficiency of our methods
depends on low-rank approximations of the preconditioned
parameter-to-observable map, which rely on properties of the
associated forward and observation operators. As a consequence of
ill-posedness, many parameter-to-observable operators admit such
low-rank approximations. Finally, our relaxation of the design problem
using a continuous weight vector in combination with sparsification
provides just indirect control of the number of sensors through the
value of $\upgamma$. However, this approach makes the combinatorial
OED problem of optimal sensor placement computationally tractable.

Possible extensions of the present work include consideration of (1)
alternative optimal experimental design criteria that are meaningful
in infinite dimensions and (2) nonlinear parameter-to-observable
maps. OED with nonlinear parameter-to-observable maps is particularly
challenging as, in general, the posterior is non-Gaussian, the
(linearization of the) parameter-to-observable map in general depends
on the state, parameter, and data variables and thus its low-rank
approximation cannot be computed a priori, optimal designs might not
be unique even if the sparsifying penalty is convex, and the misfit
Hessian depends on (usually unavailable) observations.

\section*{Acknowledgments}
We would like to thank James Martin for providing us with an
implementation of the algorithm in \cite{ChenAnitescuSaad11},
which was used to compute the application of the
inverse square root of the mass matrix.

\appendix 
\section{The mass-weighted trace estimator}\label{appdx:trace}
The following basic result justifies the form
of the randomized trace estimator for an $\M$-symmetric
matrix as defined in section~\ref{sec:bayes-disc}. The proof given below adapts the arguments
given in~\cite{AvronToledo11} regarding Gaussian trace estimators for
symmetric linear mappings on the standard Euclidean inner product space.

\begin{proposition}
Let $\A$ be an $\M$-symmetric linear mapping on $\R^n_\M$, $n\ge 1$.
Suppose $\vec{y}$ is a random $n$-vector with i.i.d.~$\GM{0}{1}$-entries, and let
$\vec{z} = \M^{-1/2}\vec{y}$.
Then, $T(\A) := \mip{\vec{z}}{\A\vec{z}}$ is an unbiased estimator for $\trace(\A)$; that is,
$\ave{T(\A)} = \trace(\A)$.
\end{proposition}
\begin{proof}
Since $\A$ is $\M$-symmetric, it admits a spectral decomposition,
$\A = \V \matrix{\Lambda} \V^*$, with $\V$ a matrix with $\M$-orthogonal
eigenvectors of $\A$ as its columns (i.e., $\V^T \M \V = \I$) 
and $\matrix{\Lambda}$ a diagonal matrix with real eigenvalues $\{\lambda_i\}_{i = 1}^n$ of $\A$ on 
its diagonal. Then,
\[
\begin{aligned}
    \ave{T(\A)} %
    = \ave{ \vec{z}^T \M \V \matrix{\Lambda} \V^* \vec{z} } %
    = \ave{ \vec{y}^T \M^{1/2} \V \matrix{\Lambda} \V^T \M^{1/2} \vec{y} } %
    = \ave{ \vec{q}^T \matrix{\Lambda} \vec{q}},
\end{aligned}
\]
with $\vec{q} = \V^T \M^{1/2}\vec{y}$; note that above we also used $\V^* = \V^T \M$. 
It is straightforward to show that $\vec{q} \sim \GM{\vec{0}}{\I}$.
Therefore, 
\[
    \ave{T(\A)} = \ave{ \vec{q}^T \matrix{\Lambda} \vec{q}} = \sum_{i = 1}^n \lambda_i \ave{q_i^2} = \sum_{i = 1}^n \lambda_i = \trace(\A),
\]
where the penultimate equality follows from the fact that $\ave{q_i^2}
= 1$, which is the case because $q_i^2$ is the square of standard
normal random variable and is thus $\chi^2$ distributed with one degree of
freedom.
\end{proof}

\newcommand{\dpard}{\dparpost^\obs}
\newcommand{\avey}[1]{\mathsf{E}_{\obs|\dpar} \left\{{#1} \right\}}
\newcommand{\avemu}[1]{\mathsf{E}_{\priorm}\left\{{#1} \right\}}
\section{Relation between the trace of posterior covariance and the expected MSE of the posterior mean}\label{appdx:MSE}
Here, we show that
\begin{equation}\label{equ:bayesrisk}
   \int_{\R^n} \int_{\R^q} \norm{\dpard - \dpar}^2_\M \, d\mu_{\obs | \dpar}(\obs) \, d\priorm(\vec{m})
   = \trace(\postcov),
\end{equation}
where, according to our choice of the noise model, $\mu_{\obs | \dpar}
= \GM{\FF \dpar}{\ncov}$ and $\priorm$ is the (discretized) prior
measure over $\R^n_\M$, $\priorm = \GM{\dpar_0}{\prcov}$.
While~\eqref{equ:bayesrisk} is known~\cite{ChalonerVerdinelli95} in the context of 
inference problems in $n$-dimensional Euclidean space, below we include a proof
of this result adjusted to our choice of mass-weighted inner product, which results from a consistent discretization 
of the infinite-dimensional Bayesian inverse problem. We use the 
notation $\dparpost = \dpard$ to make the dependence of the posterior mean to 
data explicit.  For brevity, we write~\eqref{equ:bayesrisk} as
\begin{equation}\label{equ:bayesrisk2} 
   \avemu{\avey{\norm{\dpard - \dpar}^2_\M}}    = \trace(\postcov).
\end{equation}
Suppose $\mu$ is a Gaussian measure with mean $\vec{z}_0$ and covariance $\mat{Q}$
defined on a Hilbert space $V$. 
In the proof of the following result we use the fact that $\int_V \norm{\vec{\xi}}^2 \, d\mu(\vec{\xi}) = \trace(\mat{Q}) + \norm{\vec{z}_0}^2$. 
Moreover if $\vec{z}$ is Gaussian random variable with probability law $\mu$ and $\mat{A}:V \to W$ is a linear transformation with $W$
a Hilbert space, then $\mat{A}\vec{z} + \vec{b}$ is also Gaussian
with law $\GM{\mat{A}\vec{z}_0 + \vec{b}}{\mat{A} \mat{Q} \mat{A}^*}$, 
where $\mat{A}^*$ is the adjoint of $\mat{A}$.

\begin{proposition}
Consider a Bayesian linear inverse problem with a Gaussian prior $\priorm = \GM{\dpar_0}{\prcov}$ on 
$\R^n_\M$ with additive Gaussian noise model as described above. Then, \eqref{equ:bayesrisk2} holds.
\end{proposition}
\begin{proof}
An elementary calculation shows that,
\begin{multline}\label{equ:MSE_decomp}
  \avey{ \norm{\dpard - \dpar}^2_\M} = \\ 
     \avey{\norm{\dpard - \avey{\dpard}}^2_\M} + 
     \norm{\avey{\dpard} - \dpar}^2_\M,
\end{multline}
Let us consider the first term on the right-hand-side of~\eqref{equ:MSE_decomp}; denote 
$S\obs = \dpard - \avey{\dpard}$.
Recalling that $\obs | \dpar \sim \GM{\FF\dpar}{\ncov}$, we note that for any fixed $\dpar$,
we have $S\obs = \postcov \FF^* \ncov^{-1} (\obs - \FF \dpar)$.
Therefore, the law of the random variable
$S:(\R^q, \mu_{\obs | \dpar}) \to \R^n_\M$ is $\mu_S = \GM{\vec{0}}{\postcov\HM\postcov}$, where $\HM = \FF^* \ncov^{-1} \FF$.
Therefore,
\begin{equation}\label{equ:1st_term}
     \avey{\norm{\dpard - \avey{\dpard}}^2_\M} = 
     \int_{\R^n} \norm{\vec{\xi}}^2_\M \, d\mu_S(\vec{\xi}) =
\trace(\postcov^2 \HM),
\end{equation}
where we used  that $\trace(\postcov\HM\postcov) = \trace(\postcov^2 \HM)$.
Next, we consider the expectation over $\priorm$ of the second term in~\eqref{equ:MSE_decomp}. 
We let $T:(\R^n_\M, \priorm) \to \R^n_\M$ be defined by $T\dpar = \dpar - \avey{\dpard}$.
Then,
\begin{multline}
     T\dpar %
                    = \postcov( \postcov^{-1} - \HM) \dpar - \postcov\prcov^{-1} \dpar_0
                    = \postcov \prcov^{-1}(\dpar - \dpar_0).
\end{multline}
Hence, the law of $T$ is given by $\mu_T = \GM{\vec{0}}{\mat{C}}$ with
$\mat{C} = (\postcov \prcov^{-1})\prcov(\postcov \prcov^{-1})^* = \postcov \prcov^{-1} \postcov$.
Therefore, 
\begin{equation}\label{equ:2nd_term}
 \avemu{\norm{\avey{\dpard} - \dpar}^2_\M } = 
 \int_{\R^n} \norm{\vec{\xi}}^2_\M \, d\mu_T(\vec{\xi})%
 = \trace(\postcov^2 \prcov^{-1}).
\end{equation}
Hence, using~\eqref{equ:MSE_decomp},~\eqref{equ:1st_term}, and~\eqref{equ:2nd_term}, we have
\begin{multline*}
    \avey{ \norm{\dpard - \dpar}^2_\M} = \trace(\postcov^2 \HM) + \trace(\postcov^2 \prcov^{-1})\\
                                          = \trace\big(\postcov^2(\HM + \prcov^{-1}) \big)
                                          = \trace(\postcov),  
\end{multline*}
where, for the last equality we used $\HM + \prcov^{-1} = \postcov^{-1}$.
\end{proof}

\end{document}